\documentclass[%
 amsmath,
 pra,
 a4paper,
 twocolumn,
 superscriptaddress
]{revtex4-2}

\usepackage[utf8]{inputenc}  
\usepackage[T1]{fontenc}     
\usepackage[american]{babel}  
\usepackage[colorlinks,citecolor=blue,linkcolor=blue,urlcolor=blue]{hyperref}  
\usepackage{graphicx} 
\usepackage{amsmath,amsthm,dsfont} 
\usepackage{braket,,amssymb}
\usepackage[normalem]{ulem}
\usepackage{graphicx,xcolor}
\usepackage[caption=false]{subfig}
\newtheorem{theorem}{Theorem}

\newtheorem{condition}[theorem]{Condition}

\newtheorem{observation}[theorem]{Observation}

\DeclareMathOperator{\tr}{tr}
\newcommand{\id}{\mathds{1}}

\newcommand{\nc}{\mathcal{N}_C}

\begin{document}
\title{
Fate of multiparticle entanglement when one particle becomes classical
}

\author{Zhen-Peng Xu}%
\email{zhen-peng.xu@uni-siegen.de}
\affiliation{School of Physics and Optoelectronics Engineering, Anhui University, 230601 Hefei, People’s Republic of China}
\affiliation{Naturwissenschaftlich-Technische Fakult\"at, Universit\"at Siegen, Walter-Flex-Stra{\ss}e 3, 57068 Siegen, Germany}

\author{Satoya Imai}
\email{satoyaimai@yahoo.co.jp}
\affiliation{Naturwissenschaftlich-Technische Fakult\"at, Universit\"at Siegen, Walter-Flex-Stra{\ss}e 3, 57068 Siegen, Germany}

\author{Otfried Gühne}%
\email{otfried.guehne@uni-siegen.de}
\affiliation{Naturwissenschaftlich-Technische Fakult\"at, Universit\"at Siegen, Walter-Flex-Stra{\ss}e 3, 57068 Siegen, Germany}

\date{\today}
\begin{abstract}
We study the change of multiparticle entanglement if one particle becomes classical, in the 
sense that this particle is destructed by a measurement, but the gained information is encoded 
into a new register. We present an estimation of this change for different entanglement measures 
and ways of encoding.
We first simplify the numerical calculation to analyze the change of entanglement under 
classicalization in special cases. Second, we provide general upper and lower bounds on
the entanglement change. Third, we show that the entanglement change caused by classicalization 
of one qubit only can still be arbitrarily large. Finally, we discuss cases where no entanglement 
is left under classicalization for any possible measurement. Our results shed light on the 
storage of quantum resources and help to develop a novel direction in the field of quantum 
resource theories.
\end{abstract}

\maketitle


\section{Introduction}
Different types of quantum resources~\cite{chitambar2019quantum} are 
essential for quantum information tasks, like quantum computation 
\cite{divincenzo1995quantum}, quantum key distribution \cite{scarani2009security}, 
and quantum metrology \cite{giovannetti2006quantum}, where they can provide 
a decisive advantage over the classical regime. One main problem for many 
quantum resources is their sensitivity to the disturbance from the 
environment. Their protection with tools like quantum error 
correction~\cite{lidar2013quantum} is usually expensive, especially if 
larger systems are considered. In practice, some fraction of the particles
of a larger quantum system can inevitably become classical, e.g., caused 
by a measurement or decoherence process. In fact, the particles 
may even be completely lost. 

It is a natural question to ask how multiparticle entanglement \cite{horodecki2009quantum, guhne2009entanglement} is affected by 
such processes. Many
works have considered the influence of decoherence on multiparticle entanglement~\cite{simon2002robustness,dur2004stability,carvalho2004decoherence,hein2005entanglement,guhne2008multiparticle,aolita2008scaling}.
Other works considered the robustness of multiparticle entanglement under particle
loss~\cite{briegel2001persistent,brunner2012persistency,neven2018entanglement,luo2021robust}.  
Moreover, the sharp change of bipartite entanglement caused by the complete loss of one particle in 
one party has been studied as the concept of lockable
entanglement~\cite{horodecki2005locking,christandl2005uncertainty,leung2009survey, yang2009squashed}.
There can, however, still be information left in the environment after loss 
of {particles}. For example, in the case of the Stern-Gerlach experiment, 
the left information is given by the location of the spots on the screen. 
As another example one can consider the decay of particles due to decoherence, 
where it may be reasonable to gather some information from the particles before their complete decay. The usefulness of this classical information has been extensively explored in the form of the entanglement of
assistance~\cite{divincenzo1998entanglement}, where a third party (Charly) 
optimizes the measurement and the resulting information to assist the two 
original parties (Alice and Bob) to reveal as much quantum entanglement as 
possible. Most research on the entanglement of assistance has focused on the 
case where the global state is
pure~\cite{li2010evolution,smolin2005entanglement,gour2005deterministic}. 
As it turns out~\cite{divincenzo1998entanglement}, the entanglement of 
assistance depends only on the reduced state for Alice and Bob, and the exact three-partite initial state is not important.

\begin{figure}[t!]
  \centering
  \includegraphics[width=0.45\textwidth]{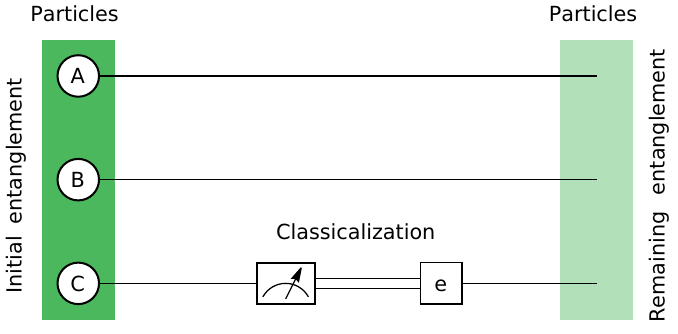}
  \caption{The change of multiparticle entanglement  if the particle $C$ 
  becomes classical. In this process of classicalization the particle $C$ 
  is first {destroyed} by the measurement and then the measurement information 
  is encoded in a new register. This paper asks for which classicalization 
  procedure the change of entanglement is minimal.}
  \label{fig:scenario}
\end{figure}

In this paper we consider a {different scenario}: One or more particles 
in a multiparticle system is destructed by a measurement. The gained classical 
information is then encoded in a quantum state. Our question is how much  
the multiparticle entanglement is affected in this process of classicalization, 
see also Fig.~\ref{fig:scenario}. 
This scenario is practically relevant, as one may not have the 
perfect `assistance' when the size and performance of the 
register system are limited. Consequently, our approach can 
provide guidance for the storage of quantum entanglement 
robust to particle loss and for finding the optimal strategy 
of entanglement recovery with the gained classical 
information and a small register system.
In comparison with the concept of quantum assistance,  we consider mixed quantum 
states where the entanglement is stored and it is not the aim of the measured 
party to increase the bipartite entanglement between the remaining ones. 
Most importantly, the initial quantum state plays a major role in the change 
of the entanglement due to classicalization.
We stress that there are further related concepts. The so-called hidden entanglement~\cite{d2014hidden} has been introduced as the difference 
between the entanglement without the decomposition information of a mixed 
state and the one with the decomposition information.  Besides, the role 
of one particle in the change of entanglement has also been considered 
in distributed entanglement~\cite{chuan2012quantum,streltsov2012quantum}, 
where the particle is transferred from one party to another one rather 
than it is destroyed.

\section{Notations and definitions}\label{sec:notations}
We focus on tripartite {systems} in this paper, other multipartite {systems}
can be analyzed similarly. We denote the {initial state} as $\rho_{ABC}$. First, 
suppose that one party of this state {is} measured in a process that completely 
destroys the measured party, such as the detection of the photon polarization.

Without loss of generality, we here assume that the destructive measurement 
${M}=\{m_i\}$ acts on the party $C$. After the measurement, the particles 
belonging to party $C$ {vanish}, but the post-measurement information from 
the associated outcome is available. That is, {each classical outcome $i$} 
can be encoded into a new register system $E$ as associated post-measurement 
states $\tau_i$. 
We say that this encoding is perfect, if 
$\tau_i = \ket{i}\!\bra{i}$ for an orthogonal basis $\{\ket{i}\}$. In 
practice, of course, the encoding may not be perfect due to the interaction 
with the environment.

{
We can write the above process as the operation
\begin{align}
\label{eq:operation}
\Phi_C(\rho_{ABC})= \sum_i p_i \sigma_i \otimes \tau_i,
\end{align}
where $p_i = \tr(\rho_{ABC} m_i)$, $\sigma_i = \tr_C(\rho_{ABC} m_i)/p_i$ and $\tau_i$ is the register state related to the outcome $i$.
We say that this encoding is {perfect}, if 
$\tau_i = \ket{i}\!\bra{i}$ for an orthogonal basis $\{\ket{i}\}$. In 
practice, of course, the encoding may not be perfect due to the interaction 
with the environment or the limited memory of the register.
}

We denote by $\mathcal{N}_C$ the set of {all possible operations in the form in Eq.~\eqref{eq:operation} on the party $C$.} We stress that the set $\mathcal{N}_C$ is equivalent to the set of entanglement breaking channels~\cite{horodecki2003entanglement} acting on the party $C$. So far, we have not imposed any assumption on the destructive measurements and the encoding, but in practice, there can be extra limitations on them.

Our central question is how much the global entanglement in $\rho_{ABC}$ is changed by the operation $\Phi_C$. The maximal change happens usually when there is no classical information left or it has not been employed, that is, the $\tau_i$ 
are the same for all outcomes $i$'s, a similar question has been explored already 
under the concept of lockable entanglement~\cite{horodecki2005locking}, see more details 
in Sec.~\ref{sec:lockability}. Here we are particularly interested in the minimal amount 
of entanglement change {with remaining classical information}, {which corresponds to the optimal operation $\Phi_C$ to keep as much entanglement as possible.}

For this purpose, we define the quantity $\Delta_{\mathcal{E}}(\rho_{ABC})$ 
as
\begin{align}
  \Delta_{\mathcal{E}}(\rho_{ABC}) &= \min_{\Phi_C\in \mathcal{N}_C}
  \left\{\mathcal{E}[\rho_{ABC}] - \mathcal{E}[\Phi_C(\rho_{ABC})]
  \right\},
\end{align}
where $\mathcal{E}$ is a tripartite entanglement measure. The practical choice of $\mathcal{E}$ may depend on the quantum information task under consideration.
{For the choice} of entanglement measures, it is necessary to require that 
$\mathcal{E}$ does not increase under local operations and classical communication (LOCC)~\cite{chitambar2014everything}{, called monotonicity under LOCC}.
In this case, $\Delta_{\mathcal{E}}(\rho_{ABC})$ is always non-negative.

Two  further remarks are in order. First, if $\mathcal{E}$ is a measure of 
genuine multipartite entanglement, then $\Delta_{\mathcal{E}}(\rho_{ABC}) = 
\mathcal{E}[\rho_{ABC}]$, since $\Phi_C(\rho_{ABC})$ is always separable with 
respect to the bipatition $AB|C$ for any $\Phi_C$ and $\rho_{ABC}$. Second, 
{if we restrict the set $\mathcal{N}_C$ with limitations on measurements and 
register states, the amount of $\Delta_{\mathcal{E}}(\rho_{ABC})$ can be affected. 
One example is to consider the operations which keep the dimension of the system.}

\section{Simplification}\label{sec:simplification}
In general it is difficult to calculate $\Delta_{\mathcal{E}}(\rho_{ABC})$, due to the complexity of characterizing the set $\mathcal{N}_C$.
Here we provide a method to simplify the calculation.
By default, we assume the entanglement measure $\mathcal{E}$ is monotonic under LOCC.
Then we have:

\begin{observation}\label{ob:encdingextreme}
  If the entanglement measure $\mathcal{E}$ is convex, we only need to consider $M = \{m_i\}$ as an extremal point in the considered measurement set $\mathcal{M}$. 
  More precisely:
\begin{align}\label{eq:simplify}
  \hspace{-0.5em}\Delta_{\mathcal{E}}(\rho_{ABC})\! =\! \min_{M\in \partial \mathcal{M}} \left\{\mathcal{E}[\rho_{ABC}] \!-\! \sum_i p_i \mathcal{E}[\sigma_i\!\otimes\! |0\rangle\langle 0|]\right\},  
\end{align}
where $\partial \mathcal{M}$ is the set of extremal points in $\mathcal{M}$, $p_i = \tr(\rho_{ABC} m_i)$ and $\sigma_i = \tr_C(\rho_{ABC} m_i)/p_i$.
\end{observation}

{The proof of Observation~\ref{ob:encdingextreme} is given in Sec.~A in the Supplemental Material~\cite{supplementalmaterial}.}
The Observation shows that the actual calculation of $\Delta_{\mathcal{E}}(\rho_{ABC})$ can be reduced to the set of extremal points in $\mathcal{M}$,  which { has been well characterized in Ref.}~\cite{d2005classical}.
In the following, we will address this 
problem for two special cases.
The first case is that the party $C$ is a qubit and the measurement information from the outcomes is also registered in a qubit system $E$~\cite{ruskai2003qubit}.
For convenience, we denote by $\mathcal{N}_1$ the set of those operations, which is equivalent to the set of all entanglement breaking channels mapping qubit to qubit.
The second case is that the measurement ${M}$ is a dichotomic POVM~\cite{d2005classical}, where $C$ is not necessarily a qubit.
We denote this set as $\mathcal{N}_2$.

Now we can present the following observation:
\begin{observation}\label{ob:dichotomic}
For a convex entanglement measure $\mathcal{E}$, {if we replace $\mathcal{N}_C$ by $\mathcal{N}_1$ or $\mathcal{N}_2$ in the definition of $\Delta_{\mathcal{E}}$}, then the value of $\Delta_{\mathcal{E}}(\rho_{ABC})$ can be achieved with projective measurements.
\end{observation}

{The proof of Observation~\ref{ob:dichotomic} is given in Sec.~B in the Supplemental Material~\cite{supplementalmaterial}.}
Observation~\ref{ob:encdingextreme} and Observation~\ref{ob:dichotomic} make the numerical calculation possible with only few parameters as in the following examples.

\subsection{Example: three-qubit systems}
Here we look at three-qubit systems and analyze $\Delta_{\mathcal{E}}(\rho_{ABC})$ 
{with $\mathcal{N}_1$ and $\mathcal{N}_2$.} Important examples of multipartite 
entanglement measures that satisfy {convexity} and monotonicity under LOCC are the 
multipartite negativity~\cite{sabin2008classification} and multipartite squashed entanglement~\cite{yang2009squashed,christandl2004squashed}:
\begin{align}
    N_{ABC}(\rho_{ABC})&= N_{AB|C} +N_{BC|A} +N_{AC|B},\\
    E_{sq}(\rho_{ABC})&=\min_{\gamma_{ABCX}} \frac{1}{2} I(A:B:C|X).
\end{align}
Here, $N_{X|Y}= \left|\sum_{\lambda_i<0}\lambda_i\right|$ is the negativity 
for a bipatition $X|Y$ with eigenvalues $\lambda_i$ of the partial transposed 
state $\rho^{T_Y}$ with respect to the subsystem $Y$, where $Y=A,B,C$.
Also, $I(A:B:C|X)=S(AX)+S(BX)+S(CX)-S(ABCX)-2S(X)$ is the quantum conditional 
mutual information, where $\gamma_{ABCX}$ is any extension of $\rho_{ABC}$, 
i.e., $\rho_{ABC} = \tr_X[\gamma_{ABCX}]$, and $S(M)$ is the von Neumann entropy 
of system $M$. For a pure state $\rho_{ABC}$, the quantum conditional mutual 
information can be simplified as $I(A:B:C|X)=S(A)+S(B)+S(C)$, which is 
independent of system $X$. 

As the first example, we consider the superposition of Greenberger-Horne-Zeilinger (GHZ) states 
and W states:
\begin{align} \label{eq:pureghzestate}
   \ket{\psi(p)}=\sqrt{p} \ket{\rm GHZ} + \sqrt{1-p} \ket{\rm W}, 
\end{align}
where
$0\leq p \leq 1$, 
$\ket{\rm GHZ} = (\ket{000}+\ket{111})/\sqrt{2}$, and
$\ket{\rm W} = (\ket{001}+\ket{010}+\ket{100})/\sqrt{3}$.
The numerical relation between $\Delta_\mathcal{E}$ and $p$ is presented 
in Fig.~\ref{fig:pureghzWW} for $\mathcal{E} = N_{ ABC}, \, E_{sq}$, 
details about the optimization method are given in in Sec.~C in the Supplemental Material~\cite{supplementalmaterial}.
Interestingly, we find that the maximal value of $\Delta_{\mathcal{E}}(\ket{\psi})$ 
is given by the W state, while the minimal value is not achieved by the GHZ state 
but the state at $p = 0.4$. {We remark that both of $N_{ABC}(\ket{\psi(p)})$ 
and $E_{sq}(\ket{\psi(p)})$ are minimized when $p=0.4$. However, it is an open problem
to understand why this state should also have minimal entanglement change.}

{
Moreover, let us consider a three-qutrit case
and compute the tuple of
$\Delta_\mathcal{E}$ for 
$\mathcal{E} = (N_{ABC}, \, E_{sq})$.
The GHZ state 
$\sum_{i=0}^2\ket{iii} /\sqrt{3}$
has $(1.667,\, 0.792489)$,
while the state 
$(
\ket{012}+\ket{120}+\ket{201}
+\ket{021}+\ket{210}+\ket{102}
)/\sqrt{6}$
has $(1.86747,\, 0.971332)$. More deitals are in Sec.~C in the Supplemental Material~\cite{supplementalmaterial}.

}
\begin{figure}[tpb]
  \centering
  \includegraphics[width=0.45\textwidth]{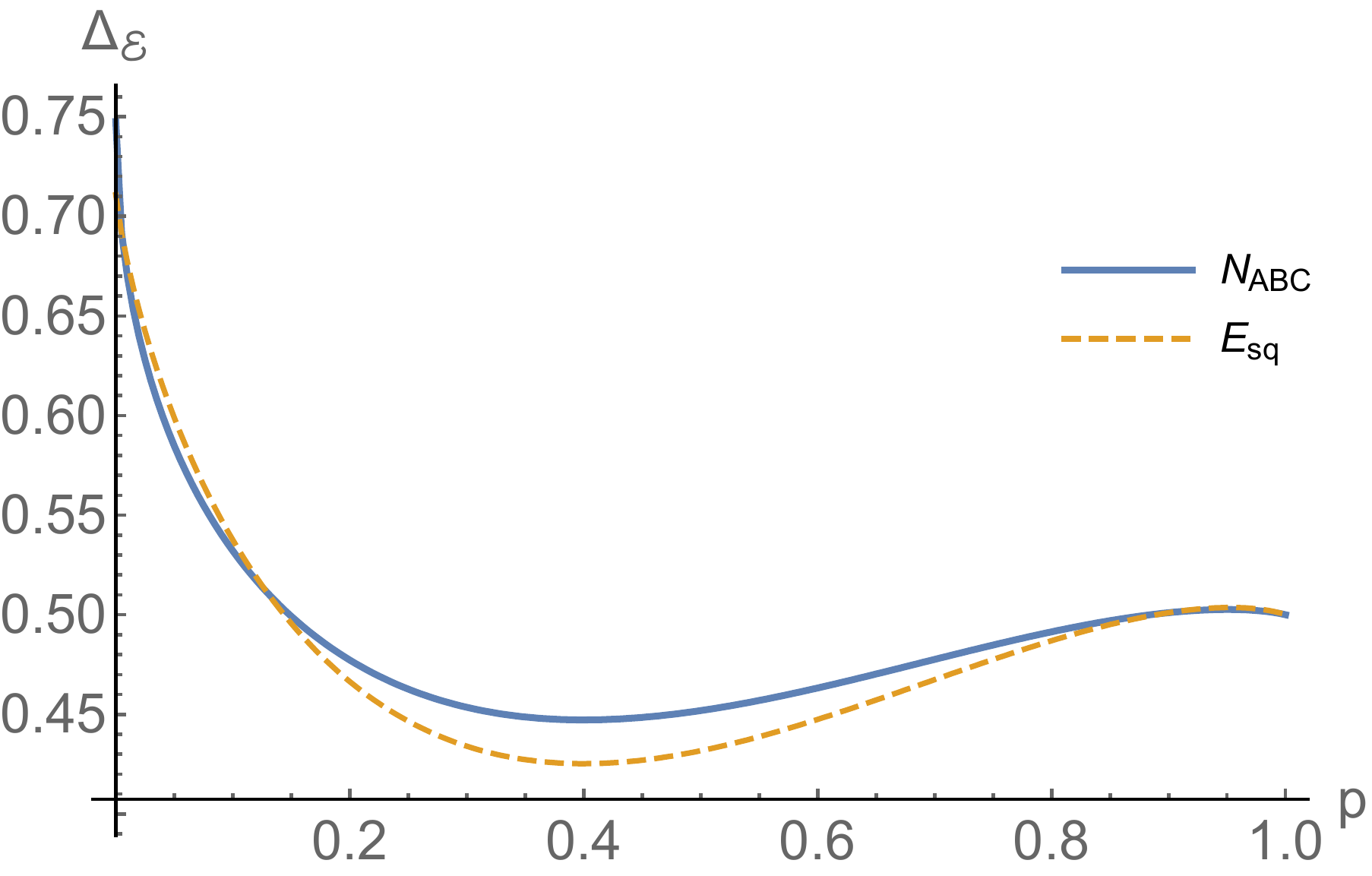}
  \caption{$\Delta_{\mathcal{E}}$ with $N_{ABC}$ and $E_{sq}$ for $\ket{\psi(p)}=\sqrt{p} \ket{\rm GHZ} + \sqrt{1-p} \ket{\rm W}$.}
  \label{fig:pureghzWW}
\end{figure}

\section{General bounds}\label{sec:general}
{In general, it may be hard to obtain the exact value of $\Delta_{\mathcal{E}}(\rho_{ABC})$ 
for some entanglement measure $\mathcal{E}$. To address this situation, we now derive upper
and lower bounds that can be useful for the estimation.} First, we present a general lower 
bound.
\begin{observation} \label{ob:lower}
 For a convex entanglement measure $\mathcal{E}$, {and for the set $\mathcal{N}_C$},
  we have
  \begin{equation}\label{ob:lowerbound}
    \Delta_{\mathcal{E}}(\rho_{ABC})
    \ge
    \min_{|x\rangle}
    \left\{
    \mathcal{E}[\rho_{ABC}]- \mathcal{E}[\sigma_{|x\rangle} \otimes |0\rangle\langle 0|] \right\},
 \end{equation}
  where $|x\rangle$ is a measurement direction on the party $C$ and $\sigma_{|x\rangle} =
  \langle x|\rho_{ABC}|x\rangle / \tr[\langle x|\rho_{ABC}|x\rangle]$ 
  is a normalized state.
\end{observation}
The proof of Observation~\ref{ob:lower} is given in Sec.~D in the Supplemental Material~\cite{supplementalmaterial}. This lower bound can be used to
characterize the complete entanglement loss, 
as we will see later in Sec.~\ref{sec:complete}.

Furthermore, suppose that we remove all the classical information of the measurement 
outcomes, that is, we encode all the measurement outcomes into the same state 
$|0\rangle$. Then we find an upper bound:
\begin{align}
    \Delta_{\mathcal{E}}(\rho_{ABC}) \le \tilde{\Delta}_{\mathcal{E}}(\rho_{ABC}),
\end{align}
for any convex entanglement measure $\mathcal{E}$, where
\begin{equation}\label{eq:uppberbound}
    \tilde{\Delta}_{\mathcal{E}}(\rho_{ABC}) = \mathcal{E}[\rho_{ABC}] - \mathcal{E}[\rho_{AB}\otimes |0\rangle\langle 0|],
\end{equation}
with $\rho_{AB} = \tr_C(\rho_{ABC})$. {We remark that $\tilde{\Delta}_{\mathcal{E}}(\rho_{ABC})$ is the maximal entanglement change, since we can always map any encoding into the state $|0\rangle\langle 0|$ with a local operation on the system $C$.}

{Let us compare $\Delta_{\mathcal{E}}$ with its lower and upper bounds using the tripartite 
negativity $N_{ABC}$. Figs.~\ref{fig:lowerboundpure} and \ref{fig:lowerboundmixed} illustrate 
the cases of the pure three-qubit state $\ket{\psi(p)}$ in Eq.~(\ref{eq:pureghzestate}) and the 
mixed three-qubit state $\rho(q) = q \rho_{\rm GHZ} + (1-q)\rho_{\rm W}$, where 
$\rho_{\rm GHZ}=\ket{\rm GHZ}\!\bra{\rm GHZ}$ and $\rho_{\rm W}=\ket{\rm W}\!\bra{\rm W}$.}
We find that the lower bound is {relatively close to $\Delta_{\mathcal{E}}$}, especially {if 
the state approximates the GHZ state.} The gap between $\Delta_{\mathcal{E}}$ and 
$\tilde{\Delta}_{\mathcal{E}}$ shows that the post-measurement information is 
more relevant for the GHZ state than for the W state.

\begin{figure}[tpb]
  \centering
  \includegraphics[width=0.45\textwidth]{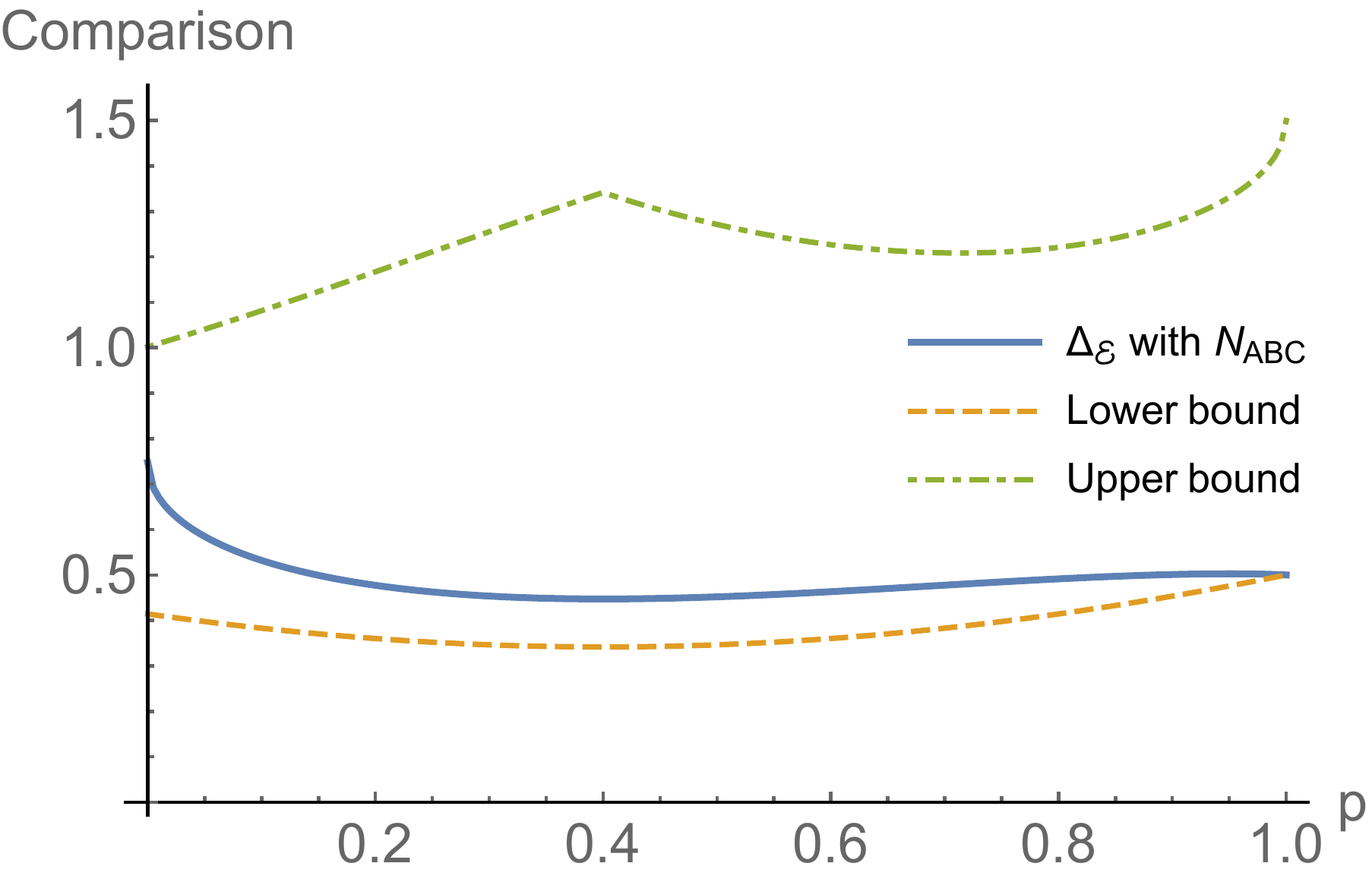}
  \caption{Comparison between $\Delta_{\mathcal{E}}$ with $N_{ABC}$ and its lower and upper bounds for the state $\ket{\psi(p)}=\sqrt{p} \ket{\rm GHZ} + \sqrt{1-p} \ket{\rm W}$.}
  \label{fig:lowerboundpure}
\end{figure}

\begin{figure}[tpb]
  \centering
  \includegraphics[width=0.45\textwidth]{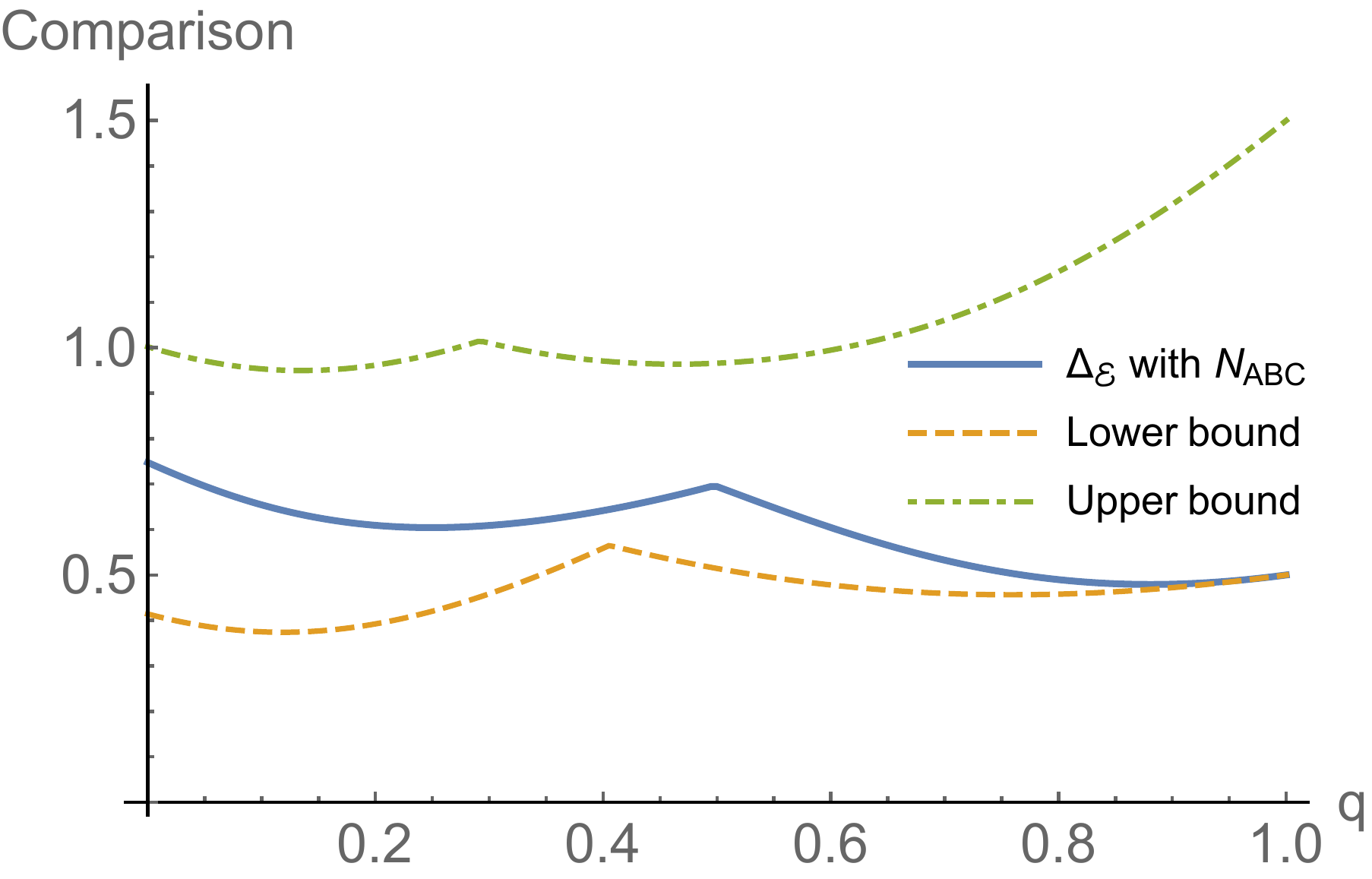}
  \caption{Comparison between $\Delta_{\mathcal{E}}$ with $N_{ABC}$ and its lower and upper bounds presented for the state $\rho(q) = q \rho_{\rm GHZ} + (1-q)\rho_{\rm W}$.}
  \label{fig:lowerboundmixed}
\end{figure}

Next, let us connect {entanglement change} to quantum discord.
For that, we consider the multipartite relative entropy of entanglement, which is the sum of the relative entropies of entanglement~\cite{linden1999reversibility} for all bipartitions, i.e., 
\begin{align}
     R_{ABC}(\rho_{ABC}) = R_{AB|C} + R_{BC|A} + R_{AC|B},
\end{align}
where $R_{X|Y} =\min_{\sigma\in {\rm SEP}} S(\rho_{XY}||\sigma)$ is the relative entropy of entanglement for a bipatition $X|Y$,
$S(\rho||\sigma)=\tr[\rho \, (\log{\rho}-\log{\sigma})]$ is the von Neumann relative entropy
and  ${\rm SEP}$ is the set of {bipartite} separable states.

Similarly, the amount of quantum discord~\cite{modi2012classical} can be also measured 
by the relative entropy: 
$D_{XY}(\rho_{XY}) = \min_{\rho'\in \Lambda}S(\rho_{XY}||\rho')$,
where $\Lambda$ is the set of quantum-classical states
$\rho'=\sum_{i} p_i \sigma_i \otimes \ket{i}\!\bra{i}$ with orthonormal basis $\{\ket{i}\}$.
Now we can formulate the following two Observations:

\begin{observation}\label{ob:bounds}
{For the entanglement measure $\mathcal{E}$ being } the tripartite relative entropy of entanglement $R_{ABC}$, we have
 \begin{align}
     R_{AB|C}(\rho_{ABC}) \leq \Delta_{\mathcal{E}}(\rho_{ABC})
     \leq 3 D_{AB|C}(\rho_{ABC}).
 \end{align}
\end{observation}

\begin{observation}\label{ob:relative}
More generally, if $D_{AB|C}(\rho_{ABC}) = 0$, then we have  $\Delta_{\mathcal{E}}(\rho_{ABC}) = 0$ for any entanglement measure $\mathcal{E}$.
\end{observation}

{The proofs of Observation~\ref{ob:bounds} and Observation~\ref{ob:relative} are given in Sec.~E and Sec.~F in the Supplemental Material~\cite{supplementalmaterial}.}
From Observation~\ref{ob:relative}, the condition $D_{AB|C}(\rho_{ABC}) = 0$ is a sufficient condition for $\Delta_{\mathcal{E}}(\rho_{ABC}) = 0$ for any measure $\mathcal{E}$.
On the other hand, this is not a necessary condition.
For instance, if the initial state $\rho_{ABC}$ is fully separable, clearly $\Delta_{\mathcal{E}}(\rho_{ABC}) = 0$, but this does not mean $D_{AB|C}(\rho_{ABC})=0$.
From the conceptional perspective, quantum discord is the difference of quantum correlation before and after a projective measurement, whereas $\Delta_{\mathcal{E}}(\rho_{ABC})$ {quantifies} the difference of entanglement, which is only one sort of quantum correlations.

\section{Lockability}\label{sec:lockability}
Previous works~\cite{horodecki2005locking,christandl2005uncertainty,leung2009survey} 
have studied a similar issue under the name of lockability of entanglement measures.  
There, one asks for the quantitative change of entanglement by the loss of one particle, 
(e.g., one qubit) {\it within} one party. For example, in the bipartite scenario, one
considers the situation where Alice and Bob have both five qubits and then one asks
how the entanglement changes if Alice looses one of her qubits. If the entanglement 
change can be arbitrarily large, the entanglement measure is {called} lockable.
For instance, all convex entanglement measures are known to be lockable, while 
the relative entropy of entanglement is not~\cite{horodecki2005locking}.

{The lockable entanglement is related to our consideration in the following sense. 
For a given triparite state $\rho_{ABC}$, if we choose the convex entanglement 
measure $\mathcal{E}$ to only measure the entanglement between the bipartition 
$A|BC$ (or $AC|B$), then $\tilde{\Delta}_{\mathcal{E}}$ defined in 
Eq.~\eqref{eq:uppberbound} is the quantity considered in lockable entanglement.
More precisely, for any convex entanglement measure $\mathcal{E}$ for the 
bipartition $A|BC$, we have
\begin{equation}
    \tilde{\Delta}_{\mathcal{E}}(\rho_{ABC}) = \mathcal{E}[\rho_{ABC}] - 
    \mathcal{E}[\rho_{AB}],
\end{equation}
where we used that $\mathcal{E}[\rho_{AB}\otimes \ket{0}\!\bra{0}]=\mathcal{E}[\rho_{AB}]$, 
see Theorem 2 in Ref.~\cite{horodecki2005simplifying}.}

In order to understand the difference between the behaviour of entanglement under classicalization and the lockability problem, one has to analyze the role of the information
coming from the measurement results. We know already from Fig.~\ref{fig:lowerboundpure} and \ref{fig:lowerboundmixed} that this information makes some difference for the entanglement change. In the following, we will show that this difference can be arbitrarily large.

\subsection{Example: Flower state}
{First, let us consider} the so-called flower state on $d \otimes d\otimes 2$-dimensional systems~\cite{christandl2005uncertainty}:
  \begin{align}
      \omega_{ABC}
      = &\frac{2}{d(d+1)}  P^{(+)}_{AB}\otimes \frac{d+1}{2d} |0\rangle\langle 0|_C\nonumber\\
      \quad
      &+  \frac{2}{d(d-1)} P^{(-)}_{AB}\otimes \frac{d-1}{2d} |1\rangle\langle 1|_C,
  \end{align}
  where $P^{(\pm)}_{AB}$ are the projections onto the symmetric and anti-symmetric 
  subspaces, that is
  $P^{(\pm)}_{AB} = (\id_{AB}\pm V_{AB})/2$ with the SWAP operator $V_{AB}$, acting as $V_{AB}\ket{v_A}\otimes\ket{v_B}=\ket{v_B}\otimes\ket{v_A}$.
  
Notice that,
{the quantum discord of $\omega_{ABC}$ for the bipartition $AB|C$ is $0$,}
i.e., $D_{AB|C}(\omega_{ABC})=0$.
From Observation~\ref{ob:relative}, we conclude that $\Delta_{\mathcal{E}}(\omega_{ABC}) = 0$ for any entanglement measure $\mathcal{E}$.
However, we have $\tilde{\Delta}_{\mathcal{E}}(\omega_{ABC}) = \mathcal{E}(\omega_{ABC}) > 0$ 
{, because $\tr_C(\omega_{ABC}) \otimes |0\rangle\langle 0|$ is fully separable. 
} 
In fact, if the entanglement measure $\mathcal{E}$ is taken as the squashed entanglement, then $\mathcal{E}(\omega_{ABC})$ can be arbitrarily large~\cite{christandl2005uncertainty}. This directly implies that the difference $\tilde{\Delta}_{\mathcal{E}} - {\Delta}_{\mathcal{E}}$ can be arbitrarily large {by choosing $d$ properly}.
Hence, although the  information from the measurement {at the flower state} is only one bit, 
a large amount of entanglement can be saved by collecting it.

\subsection{Example: $n$-pairs of Bell states}
{On the other hand, we will see that the entanglement change $\Delta_{\mathcal{E}}$ 
can also be arbitrarily large even if only one qubit has become classical.}
As example, let us consider a pure state made of $n$ pairs of Bell state $\ket{\Psi^+}=(\ket{00}+\ket{11})/\sqrt{2}$.
We label the $i$-th pair of particles with $a_i, b_i$.
Suppose that the party $A$ owns the particles $\{a_i\}_{i=1}^n$, the party $B$ owns the particles $\{b_i\}_{i=1}^{n-1}$, and the party $C$ owns the particle $b_n$.
We denote this state as $\beta_{ABC}=|{\Psi^+}\rangle\langle {\Psi^+}|^{\otimes n}$.
Now we can present the following {observation which is proven in Sec.~G in the Supplemental Material~\cite{supplementalmaterial}.}
\begin{observation} \label{ob:lockneg}
  For {the entanglement measure $\mathcal{E}$ to be} the tripartite negativity $N_{ABC}$, we have
  \begin{align}
    \Delta_{\mathcal{E}}(\beta_{ABC}) = 2^{n-2} + 1/2.
\end{align}
Thus, $\Delta_{\mathcal{E}}(\beta_{ABC})$ can be arbitrary large {by choosing $n$ properly}.
\end{observation}

{
Inspired by those two examples, an interesting question arises whether there exist entanglement measures $\mathcal{E}$ and states $\rho_{ABC}$ such that both $\Delta_{\mathcal{E}}(\rho_{ABC})$ and $\tilde{\Delta}_{\mathcal{E}}(\rho_{ABC}) - {\Delta}_{\mathcal{E}}(\rho_{ABC})$ can be arbitrarily large {in the sense that they are not limited by the size of $C$}, even if $C$ is only a qubit. We leave this question for further research.}

\section{Complete entanglement loss under classicalization}\label{sec:complete}
By definition, $\Delta_{\mathcal{E}}(\rho_{ABC}) \le \mathcal{E}[\rho_{ABC}]$ {always holds.}
We are now concerned about the case where this inequality is saturated, i.e., $\Delta_{\mathcal{E}}(\rho_{ABC}) =\mathcal{E}[\rho_{ABC}]$, or equivalently, 
$\max_{\Phi_C\in \mathcal{N}_C}\mathcal{E}[\Phi_C(\rho_{ABC})]=0$.

First of all, Observation~\ref{ob:lower} implies a sufficient condition for 
complete entanglement loss under classicalization, {which can be formulated as 
follows.}

\begin{condition}\label{cond1}
If, after a projective measurement in any direction $\ket{x}$ on $C$, the post-measurement state $\sigma_{\ket{x}} \propto \langle x|\rho_{ABC}|x\rangle$ is always separable, 
then the entanglement is completely lost under classicalization.
\end{condition}

Clearly, Condition \ref{cond1} is stronger than the condition that the reduced state 
$\rho_{AB}$ is separable. For instance, let us consider the GHZ state.
Its reduced state $\tr_C[\rho_{\text{GHZ}}]$ is separable, but its 
post-measurement state $\sigma_{\ket{x}}$ can be entangled if measurement 
bases are $\{\ket{+}, \ket{-}\}$.

The existence of genuine multipartite entangled states which satisfy 
Condition~\ref{cond1}, however, has already been reported in Ref.~\cite{miklin2016multiparticle}.
{We will propose observations using Condition \ref{cond1} and provide more examples in in Sec.~H and Sec.~I in the Supplemental Material~\cite{supplementalmaterial}.}


\section{Conclusion and discussion}\label{sec:conclusion}
Multiparticle quantum entanglement is an important quantum 
resource and the preservation of entanglement is 
a practical {issue}. We have studied the change of 
multiparticle entanglement under classicalization of one particle. 
Clearly, the results usually depend on the choice of the entanglement 
quantifier, and the change of entanglement is difficult to compute. We 
provided simplifications for important special scenarios and  upper 
and lower bounds for the general case. One crucial question is whether 
one small part {like one qubit} can change a lot quantum resources 
like quantum entanglement or not. Our results show that the entanglement 
change can be still arbitrarily large even with complete measurement 
information left. Besides, the measurement information can also make 
an arbitrary large difference. Finally,  we provide conditions under which
quantum entanglement is always completely lost under classicalization.

While we focused on the difference of original quantum resource and the remaining
resource if one party becomes classical, the behaviour of quantum resources 
{\it during} the quantum to classical transition is also interesting, and it may 
have a richer theoretical structure. We believe that our work paves a way to 
the design of concepts for quantum resource storage and may help to develop 
a novel direction in the field of quantum resource theories. 
\begin{acknowledgments}
We thank
H. Chau Nguyen,
Martin Pl\'{a}vala,  
Benjamin Yadin,
and
Xiao-Dong Yu
for discussions. 
This work was supported by the
Deutsche Forschungsgemeinschaft
(DFG, German Research Foundation, 
project numbers 447948357 and 440958198), 
the Sino-German Center for Research Promotion 
(Project M-0294), the ERC (Consolidator Grant No. 683107/TempoQ),
the German Ministry of Education and Research 
(Project QuKuK, BMBF Grant No. 16KIS1618K),
the Humboldt foundation, and the DAAD.

\end{acknowledgments}

\onecolumngrid
\renewcommand\thesection{\Alph{section}}

\addtocounter{theorem}{-7}
\addtocounter{section}{-7}
\section{Proof of Observation~\ref{ob:encdingextreme2}}\label{ap:a}
\begin{observation}\label{ob:encdingextreme2}
  If the entanglement measure $\mathcal{E}$ is convex, we only need to consider $M = \{m_i\}$ as an extremal point in the considered measurement set $\mathcal{M}$. 
  More precisely:
\begin{align}
  \hspace{-0.5em}\Delta_{\mathcal{E}}(\rho_{ABC})\! =\! \min_{M\in \partial \mathcal{M}} \left\{\mathcal{E}[\rho_{ABC}] \!-\! \sum_i p_i \mathcal{E}[\sigma_i\!\otimes\! |0\rangle\langle 0|]\right\},  
\end{align}
where $\partial \mathcal{M}$ is the set of extremal points in $\mathcal{M}$, $p_i = \tr(\rho_{ABC} m_i)$ and $\sigma_i = \tr_C(\rho_{ABC} m_i)/p_i$.
\end{observation}

\begin{proof}
For any entanglement-breaking channel $\Phi_C$, we have the decomposition:
\begin{align}\label{eq:operation-appendix}
    \Phi_C(\rho_{ABC})= \sum_i p_i \sigma_i \otimes \tau_i,
\end{align}
where $M = \{m_i\}$ is a measurement acting on $C$, $p_i = \tr(\rho_{ABC} m_i)$, $\sigma_i = \tr_C(\rho_{ABC} m_i)/p_i$, and $\tau_i$ is the state encoding the measurement outcome $i$.

Since the set $\mathcal{M}$ of all POVMs acting on $C$ is convex, any POVM $M = \{m_i\}$ can be decomposed into the convex combinations of extreme points of $\mathcal{M}$. That is, we have 
\begin{align}
    m_i = \sum_k c_k m_i^{(k)}, \forall i,
\end{align}
where $M^{(k)} = \{m_i^{(k)}\}$ is an extreme point in the set $\mathcal{M}$
and $0<c_k \leq 1$ with $\sum_k c_k = 1$.
Consequently, the operation $\Phi_C$ can be rewritten as
\begin{align}
    \Phi_C(\rho_{ABC}) = \sum_k c_k \Phi_C^{(k)}(\rho_{ABC}),
\end{align}
where 
\begin{equation}
     \Phi_C^{(k)}(\rho_{ABC}) = \sum_i \tr_C\left(\rho_{ABC} m_i^{(k)}\right) \otimes \tau_i.
\end{equation}
In the case that the entanglement measure $\mathcal{E}$ is convex, we have
\begin{align}
    \mathcal{E}[\Phi_C(\rho_{ABC})]
    \le \sum_k c_k \mathcal{E}\left[\Phi_C^{(k)}(\rho_{ABC})\right]
    \le \max_k \mathcal{E}\left[\Phi_C^{(k)}(\rho_{ABC})\right].
\end{align}
This implies that the maximal value of $\mathcal{E}[\Phi_C(\rho_{ABC})]$, or equivalently, the value of $\Delta_{\mathcal{E}}(\rho_{ABC})$, can always be achieved by extreme POVMs. 
That is,
\begin{equation}\label{eq:step1}
    \max_{\Phi_C\in \mathcal{N}_C} \mathcal{E}\left[\Phi_C(\rho_{ABC})\right] = \max_{M\in \partial\mathcal{M}, {\{\tau_i\}}} \mathcal{E}\left(\sum_i p_i \sigma_i \otimes \tau_i\right),
\end{equation}
where $\partial\mathcal{M}$ is the set of all extreme POVMs.

Note that, any imperfect encoding can be generated from the perfect one by local operations. Since the entanglement measure $\mathcal{E}$ is LOCC monotonic, {
we have $\mathcal{E}\left(\sum_i p_i \sigma_i \otimes \tau_i\right) \le \mathcal{E}\left(\sum_i p_i \sigma_i \otimes |i\rangle\langle i|_C\right)$.
}
This implies that, 
\begin{equation}\label{eq:step21}
    \max_{M\in \partial\mathcal{M}, {\{\tau_i\}}} \mathcal{E}\left(\sum_i p_i \sigma_i \otimes \tau_i\right) \le \max_{M\in \partial\mathcal{M}} \mathcal{E}\left(\sum_i p_i \sigma_i \otimes |i\rangle\langle i|_C\right).
\end{equation}
Since $\{\tau_i = |i\rangle\langle i|\}$ is just a special encoding, we have
\begin{equation}\label{eq:step22}
    \max_{M\in \partial\mathcal{M}, {\{\tau_i\}}} \mathcal{E}\left(\sum_i p_i \sigma_i \otimes \tau_i\right) \ge \max_{M\in \partial\mathcal{M}} \mathcal{E}\left(\sum_i p_i \sigma_i \otimes |i\rangle\langle i|_C\right).
\end{equation}
In total, we know that
\begin{equation}\label{eq:step2}
    \max_{M\in \partial\mathcal{M}, {\{\tau_i\}}} \mathcal{E}\left(\sum_i p_i \sigma_i \otimes \tau_i\right) = \max_{M\in \partial\mathcal{M}} \mathcal{E}\left(\sum_i p_i \sigma_i \otimes |i\rangle\langle i|_C\right).
\end{equation}
Besides, we have
\begin{align}\label{eq:step3}
    \mathcal{E}\left(\sum_i p_i \sigma_i \otimes \ket{i}\!\bra{i}_C\right)
    &= \mathcal{E}\left(\sum_i p_i \sigma_i \otimes (\ket{0}\!\bra{0} \otimes \ket{i}\!\bra{i})_C\right)\nonumber\\
    &= \sum_i p_i \mathcal{E}[\sigma_i \otimes \ket{0}\!\bra{0}_C],
\end{align}
{where the equalities in the first line holds since $\{|i\rangle\langle i|\}$ and $\{|0\rangle\langle 0|\otimes|i\rangle\langle i|\}$ can be converted to each other by LOCC, the equality in the second line is from the flag condition satisfied by any entanglement measure which is monotonic under LOCC,} see Theorem~2 in Ref.~\cite{horodecki2005simplifying}. 

By putting Eq.~\eqref{eq:step1}, Eq.~\eqref{eq:step2} and Eq.~\eqref{eq:step3} together, we complete the proof.
\end{proof}
{We recommend the reader to refer to Ref.~\cite{d2005classical} for more characterization of extreme POVMs, like necessary conditions and sufficient conditions.}
\section{Proof of Observation~\ref{ob:dichotomic2}}\label{ap:b}
\begin{observation}\label{ob:dichotomic2}
For a convex entanglement measure $\mathcal{E}$, {if we replace $\mathcal{N}_C$ by $\mathcal{N}_1$ or $\mathcal{N}_2$ in the definition of $\Delta_{\mathcal{E}}$}, then the value of $\Delta_{\mathcal{E}}(\rho_{ABC})$ can be achieved with projective measurements.
\end{observation}
\begin{proof}
From Observation~\ref{ob:encdingextreme2}, we know that for a convex entanglement measure $\mathcal{E}$ that satisfies the monotonicity condition, the optimal value of $\Delta_{\mathcal{E}}(\rho_{ABC})$ can always be obtained by the extreme points of destructive measurements in the sets $\mathcal{N}_1$ and $\mathcal{N}_2$.
Then it is sufficient to show that these extreme points are given by projective measurements.

{
First, we consider the case of $\mathcal{N}_1$.}
As proven in Ref.~\cite{ruskai2003qubit}, any entanglement breaking channel from qubit to qubit, i.e., any channel in $\mathcal{N}_1$, can be decomposed as a convex combination of classical-quantum channels.
Here recall that a channel $\Phi_C$ is called a classical-quantum channel if
\begin{equation}
    \Phi_C(\rho) = \sum_{i} \langle x_i|\rho|x_i\rangle \otimes \tau_i,
\end{equation}
where $\{|x_i\rangle\}$ is an orthonormal basis.
By definition, the classical-quantum channel is written in the composition of projective measurements and local state preparation.
That is, the extreme point in $\mathcal{N}_1$ is obtained by projective measurements.

{Next, we consider the case of $\mathcal{N}_2$.}
It is known that a POVM $\{m_1,\ldots,m_k\}$ is extreme if $m_i, m_j$ have disjoint supports for any $i\neq j$~\cite{d2005classical}.
In the dichotomic case, $m_1 = \id - m_2$, thus, $m_1, m_2$ can be diagonalized simultaneously.
Then, there is no overlap between the supports of $m_1, m_2$ if and only if they are orthogonal projectors.
Hence, the extremal points in $\mathcal{N}_2$ are also obtained by projective measurements.
\end{proof}

\section{Details of computation in figures}\label{ap:c}
Since we consider the set of entanglement breaking channels from qubit to qubit in the examples, we only need to focus on dichomatic projective measurements $M = \{m_0, m_1\}$ and perfect encoding of the outcomes according to Observation~\ref{ob:encdingextreme2} and Observation~\ref{ob:dichotomic2}. 
In this case we have,
\begin{align}\label{eq:simplified}
  \Delta_{\mathcal{E}}(\rho_{ABC}) =  \mathcal{E}[\rho_{ABC}] - \max_{M\in \mathcal{P}}\sum_{i=0,1} p_i \mathcal{E}[\sigma_i\!\otimes\! |0\rangle\langle 0|],  
\end{align}
where $\mathcal{P}$ is the set of all dichotomatic projective measurements on qubit $C$, $p_i = \tr(\rho_{ABC} m_i)$, and $\sigma_i = \tr_C(\rho_{ABC} m_i)/p_i$.
Here, the entanglement measure $\mathcal{E}$ is taken to be either the multipartite negativity $N_{ABC}$ or the multipartite squashed entanglement $E_{sq}$.

First, let us consider the case of the multipartite negativity $N_{ABC}$.
Then we have
\begin{align}
    N_{ABC}(\sigma_i\!\otimes\! |0\rangle\langle 0|)&= N_{AB|C}(\sigma_i\!\otimes\! |0\rangle\langle 0|) +N_{BC|A}(\sigma_i\!\otimes\! |0\rangle\langle 0|) +N_{AC|B}(\sigma_i\!\otimes\! |0\rangle\langle 0|)\nonumber\\
    &= N_{B|A}(\sigma_i) +N_{A|B}(\sigma_i)\nonumber\\
    &= 2N_{A|B}(\sigma_i),
\end{align}
where the second equality is from the fact that  $\sigma_{i}^{T_A}\otimes |0\rangle\langle 0|$ has same non-zero eigenvalues as $\sigma_{i}^{T_A}$ as well as for the case  $B$.

Second, let us consider the case of the multipartite squashed entanglement $E_{sq}$.
Note that, for any $4$-partite state $\eta_{ABCX}$ such that $\tr_X(\eta_{ABCX}) = \sigma_i\!\otimes\! |0\rangle\langle 0|$, it can only be in the form $\gamma_{ABX}\!\otimes\! |0\rangle\langle 0|$, where $\tr_X(\gamma_{ABX}) = \sigma_i$. Thus,
\begin{align}
    E_{sq}(\sigma_i\!\otimes\! |0\rangle\langle 0|)&=\min_{\gamma_{ABX}\!\otimes\! |0\rangle\langle 0|} \frac{1}{2} I(A:B:C|X)\nonumber\\
    &=\min_{\gamma_{ABX}\!\otimes\! |0\rangle\langle 0|} \frac{1}{2} [S(AX)+S(BX)+S(CX)-S(ABCX)-2S(X)]\nonumber\\
    &=\min_{\gamma_{ABX}\!\otimes\! |0\rangle\langle 0|} \frac{1}{2} [S(AX)+S(BX)+S(X)-S(ABX)-2S(X)]\nonumber\\
    &=\min_{\gamma_{ABX}} \frac{1}{2} [S(AX)+S(BX)-S(ABX)-S(X)]\nonumber\\
    &=E_{sq}^{(2)}(\sigma_i),
\end{align}
where
in the third line we employ the additivity of the von Neumann entropy, and we denote $E_{sq}^{(2)}$ the bipartite squashed entanglement~\cite{christandl2004squashed}.
In the case that $\rho_{ABC}$ is a pure state, each $\sigma_i$ is also a pure state.
From the result of Ref.~\cite{christandl2004squashed}, we have
\begin{align}
    E_{sq}^{(2)}(\sigma_i) &= S(A)+S(B).
\end{align}

Therefore, once we have parameterized the $2$-dimensional projective measurement $M$, the numerical calculation of $\Delta_{\mathcal{E}}(\rho_{ABC})$ can be easily performed by brute force optimization in each example. {To be more explicitly, each $2$-dimensional rank-$1$ projective measurement $M$ corresponds to a vector which can be parameterized as $\langle v| = (\cos x, e^{it} \sin x)$ such that $M = \{|v\rangle\langle v|, \id - |v\rangle\langle v|\}$. In the calculation, we have taken $x$ in the discrete set $\{\pi k/300\}_{k=0}^{300}$ and $t$ in the set $\{\pi j/50\}_{j=0}^{50}$. For each measurement direction defined by the pair $(x,t)$, the poset-selected bipartite states and their entanglement can be computed directly by the definition of the entanglement measure. By choosing the maximal entanglement of the post-measurement state over all pairs $(x,t)$, we obtain the numerical approximation of $\Delta_{\mathcal{E}}(\rho_{ABC})$ for $\mathcal{E}$ either to be $N_{ABC}$ or $E_{sq}$.

We remark that the three-dimensional non-trivial projective measurements can also be parameterized by $M = \{|v\rangle\langle v|, \id - |v\rangle\langle v|\}$, where $|v\rangle$ is a three-dimensional complex vector $(\cos x_1, e^{it_1} \sin x_1\cos x_2, e^{it_2} \sin x_1\sin x_2)$.
Note that for the sake of simplicity we considered the case of only real parameters to obtain the result of three-qutrit states in the main text.
}

\section{Proof of Observation~\ref{ob:lower2}}\label{ap:d}
\begin{observation} \label{ob:lower2}
 For a convex entanglement measure $\mathcal{E}$, {and for the set $\mathcal{N}_C$},
  we have
  \begin{equation}
    \Delta_{\mathcal{E}}(\rho_{ABC})
    \ge
    \min_{|x\rangle}
    \left\{
    \mathcal{E}[\rho_{ABC}]- \mathcal{E}[\sigma_{|x\rangle} \otimes |0\rangle\langle 0|] \right\},
 \end{equation}
  where $|x\rangle$ is a measurement direction on the party $C$ and $\sigma_{|x\rangle} =
  \langle x|\rho_{ABC}|x\rangle / \tr[\langle x|\rho_{ABC}|x\rangle]$ 
  is a normalized state.
\end{observation}
\begin{proof}
 For a given entanglement breaking channel $\Phi_C$, it can be equivalently characterized~\cite{horodecki2003entanglement} by a POVM with $M=\{q_i |x_i\rangle\langle x_i|\}$ and a preparation $\{|\psi_i\rangle\langle \psi_i|\}$.
 That is,
 \begin{align}
     \Phi_C(\rho_{ABC}) &= \sum_i q_i \langle x_i|\rho_{ABC} |x_i\rangle \otimes |\psi_i\rangle\langle \psi_i| \nonumber\\
     &= \sum_i q_i p_i \sigma_{|x_i\rangle} \otimes |\psi_i\rangle\langle \psi_i|,
 \end{align}
where $p_i = \tr(\langle x_i|\rho_{ABC} |x_i\rangle)$, $\sigma_{|x_i\rangle}$ is the normalized state of $\langle x_i|\rho_{ABC} |x_i\rangle$, and $\sum_{i}q_i p_i =1$.

For any convex entanglement measure $\mathcal{E}$, we then have
\begin{align}
     \mathcal{E}[\Phi_C(\rho_{ABC})]
     &\le \sum_i q_i p_i \mathcal{E}[\sigma_{|x_i\rangle} \otimes |\psi_i\rangle\langle \psi_i|] \nonumber\\
     &\le \max_i \mathcal{E}[\sigma_{|x_i\rangle} \otimes |\psi_i\rangle\langle \psi_i|]\nonumber\\
     &\le \max_{|x\rangle} \mathcal{E}[\sigma_{|x\rangle}\otimes |0\rangle\langle 0|],
 \end{align}
 where in the last line we apply local untary operations on the party $C$ to rotate the states to $\ket{0}$ and maximize over a more general range of measurement directions.
\end{proof}
{In principle, the optimization can be done similarly as in Appendix C. As for the application of Observation~\ref{ob:lower2} in Condition~7, we only need to show that $\sigma_{|x\rangle}$ is separable for each $|x\rangle$, which can be checked by the PPT condition in the case that $A$ and $B$ are two-dimensional subsystems with symbolic calculations. For this purpose, we do not need to specify the values of parameters in $|x\rangle$.}

\section{Proof of Observation~\ref{ob:bounds2}}\label{ap:e}

\begin{observation}\label{ob:bounds2}
{For the entanglement measure $\mathcal{E}$ being } the tripartite relative entropy of entanglement $R_{ABC}$, we have
 \begin{align}
     R_{AB|C}(\rho_{ABC}) \leq \Delta_{\mathcal{E}}(\rho_{ABC})
     \leq 3 D_{AB|C}(\rho_{ABC}).
 \end{align}
\end{observation}
\begin{proof}
 We begin by noting that Lemma 1 in Ref.~\cite{chuan2012quantum}:
 for a given tripartite state $\rho_{ABC}$, it holds that
\begin{equation}
    \hspace{-0.5em}R_{BC|A}(\rho_{ABC})
    \!\leq\! {D}_{AB|C}(\rho_{ABC})\! +\!  R_{BC|A}[\Phi_C(\rho_{ABC})],
\end{equation}
where $\Phi_C(\rho_{ABC})= \sum_i p_i \sigma_i^{AB} \otimes \ket{i}\!\bra{i}^C$ where $\tau_i = \ket{i}\!\bra{i}^C$. 
Exchanging $A$ and $B$, we similarly have
\begin{align}
    \hspace{-0.8em}R_{AC|B}(\rho_{ABC})
    \!\leq\! {D}_{AB|C}(\rho_{ABC})
    \!+\! R_{AC|B}[\Phi_C(\rho_{ABC})].
\end{align}
Summarizing both inequalities leads to
\begin{align}
    R_{BC|A}(\rho_{ABC}) + R_{AC|B}(\rho_{ABC})
    \leq 2{D}_{AB|C}(\rho_{ABC})
    + R_{ABC}[\Phi_C(\rho_{ABC})],
\end{align}
where
we use the fact that 
$R_{AB|C}[\Phi_C(\rho_{ABC})] = 0$
since $\Phi_C(\rho_{ABC})$ is separable with respect to $AB|C$.
Rewriting this left hand side as $R_{ABC}(\rho_{ABC})-R_{AB|C}(\rho_{ABC})$, we have
\begin{align}
    R_{ABC}(\rho_{ABC}) - R_{ABC}[\Phi_C(\rho_{ABC})]
    \leq 2{D}_{AB|C}(\rho_{ABC})
    + R_{AB|C}(\rho_{ABC}).
\end{align}
By definition, $\Delta_{\mathcal{E}}(\rho_{ABC})$ is always no more than this left hand side, {since $\Phi_C$ is just a special entanglement-breaking channel}.
Then we obtain
\begin{align}
    \Delta_{\mathcal{E}}(\rho_{ABC})
    \leq 2{D}_{AB|C}(\rho_{ABC})
    + R_{AB|C}(\rho_{ABC}).
\end{align}
Finally, since $R_{AB|C}(\rho_{ABC}) \leq {D}_{AB|C}(\rho_{ABC})$,
we find the upper bound.

Concerning the lower bound, we have
\begin{align}
    \Delta_{\mathcal{E}}(\rho_{ABC})
    &= \min_{\Phi_C \in \nc}
    \left\{R_{ABC}(\rho_{ABC})- R_{ABC}[\Phi_C(\rho_{ABC})]\right\}\nonumber\\
    &\geq
    R_{AB|C}(\rho_{ABC})
    +\min_{\Phi_C \in \nc}
    \left\{R_{BC|A}(\rho_{ABC})- R_{BC|A}[\Phi_C(\rho_{ABC})]\right\}\nonumber\\
    &+\min_{\Phi_C \in \nc}
    \left\{R_{AC|B}(\rho_{ABC})- R_{AC|B}[\Phi_C(\rho_{ABC})]\right\},
\end{align}
where we again use that $R_{AB|C}[\Phi_C(\rho_{ABC})] = 0$.
Since the relative entropy of entanglement satisfies the monotonicity condition, we have that
$R_{BC|A}(\rho_{ABC})- R_{BC|A}[\Phi_C(\rho_{ABC})] \geq 0$
and
$R_{AC|B}(\rho_{ABC})- R_{AC|B}[\Phi_C(\rho_{ABC})] \geq 0$.
Then we arrive at the lower bound.
\end{proof}

\section{Proof of Observation~\ref{ob:relative2}}\label{ap:f}

\begin{observation}\label{ob:relative2}
More generally, if $D_{AB|C}(\rho_{ABC}) = 0$, then we have  $\Delta_{\mathcal{E}}(\rho_{ABC}) = 0$ for any entanglement measure $\mathcal{E}$.
\end{observation}
\begin{proof}
We note that $D_{AB|C}(\rho_{ABC})=0$ if and only if there exists an entanglement-breaking channel $\Phi_C$ such that $\Phi_C(\rho_{ABC}) = \rho_{ABC}$ (see Proposition 21 in Ref.~\cite{seshadreesan2015fidelity} for more details).
By definition, 
{
\begin{align}
  \Delta_{\mathcal{E}}(\rho_{ABC}) &= \min_{\Phi'_C\in \mathcal{N}_C}
  \left\{\mathcal{E}[\rho_{ABC}] - \mathcal{E}[\Phi'_C(\rho_{ABC})]
  \right\},\nonumber\\ 
  &\le \mathcal{E}[\rho_{ABC}] - \mathcal{E}[\Phi_C(\rho_{ABC})] \nonumber\\
  &= \mathcal{E}[\rho_{ABC}] - \mathcal{E}[\rho_{ABC}]\nonumber\\
  &=0.
\end{align}
Since $\Delta_{\mathcal{E}}(\rho_{ABC})$ is nonnegative for any entanglement measure $\mathcal{E}$ which is monotonic under LOCC,
}
this eventually implies that $\Delta_{\mathcal{E}}(\rho_{ABC}) = 0$ for any entanglement measure $\mathcal{E}$ { which is assumed to be monotonic under LOCC}.
\end{proof}

\section{Proof of Observation~\ref{ob:lockneg2}}\label{ap:g}
\begin{observation} \label{ob:lockneg2}
  For {the entanglement measure $\mathcal{E}$ to be} the tripartite negativity $N_{ABC}$, we have
  \begin{align}
    \Delta_{\mathcal{E}}(\beta_{ABC}) = 2^{n-2} + 1/2.
\end{align}
Thus, $\Delta_{\mathcal{E}}(\beta_{ABC})$ can be arbitrary large.
\end{observation}
\begin{proof}
To prove this, we first show that for a $d\times d$-dimensional bipartite state, its negativity is no more than $(d-1)/2$.
Since the negativity is a convex function, we only need to prove it for {pure states}.
Let us write a pure state $|\psi\rangle$ as
  \begin{equation}
    \ket{\psi} = \sum_{i=1}^d \lambda_i \ket{a_ib_i},
    \,\, \sum_i \lambda_i^2 = 1,
    \,\, \lambda_i \ge 0.
  \end{equation}
  Then direct calculation yields that
  \begin{equation} \label{eq:uppernegativity}
    N(\ket{\psi})
    = \sum_{1\le i<j\le n} \lambda_i\lambda_j
    \le \frac{d-1}{2} \sum_{i=1}^d \lambda_i^2 = \frac{d-1}{2}.
  \end{equation}
  Here the maximal value $(d-1)/2$ can be saturated by the maximally entangled state
  $\ket{\Psi^+_d} = \frac{1}{\sqrt{d}}\sum_{i=0}^{d-1}\ket{ii}$.

Next, let us recall the $n$-copy of Bell state  $\beta_{ABC}=|{\Psi^+}\rangle\langle {\Psi^+}|^{\otimes n}$.
We remark that this $n$-copy state can be represented by the maximally entangled state in $(2^n\times2^n)$-dimensional systems $\ket{\Psi^+_{2^n}}$.
This leads to
\begin{align}\label{eq:g4}
    N_{BC|A}(\beta_{ABC}) = (2^n - 1) /2.
\end{align}

Suppose that an entanglement breaking channel $\Phi_C$ acts on the $n$-th particle of the last party $b_n$, equivalently, on the party $C$.
Since all entanglement breaking channels can be decomposed into measure and prepare operations, we again write the measure process for $\Phi_C$ as the form of the POVM with $M=\{q_i|x_i\rangle\langle x_i|\}$ and the preparation process as $\{|\psi_i\rangle\langle \psi_i|\}$, i.e.,
\begin{align}
     \Phi_C(\beta_{ABC})
     = \sum_i q_i p_i\sigma_{\ket{x_i}} \otimes |\psi_i\rangle\langle \psi_i|,
 \end{align}
where
$p_i = \tr(\langle x_i|\beta_{ABC}|x_i\rangle)$,
$\sigma_{\ket{x_i}}$ is the normalized pure state of $\langle x_i|\beta_{ABC}|x_i\rangle$,
and
$\sum_{i} q_i p_i =1$.
Then we have
 \begin{align}
    N_{BC|A}[\Phi_C(\beta_{ABC})]
    &=N_{BC|A}\left(\sum_i q_i p_i \sigma_{\ket{x_i}} \otimes |\psi_i\rangle\langle \psi_i|\right)\nonumber\\
    &\le\sum_i q_i p_i N_{BC|A}\left( \sigma_{\ket{x_i}} \otimes |\psi_i\rangle\langle \psi_i|\right)\nonumber\\
    &=\sum_i q_i p_i N_{BC|A}\left( \sigma_{\ket{x_i}} \otimes |0\rangle\langle 0|\right)\nonumber\\
    &= \sum_i q_i p_i N_{AB}(\sigma_{\ket{x_i}})\nonumber\\
    &\le \sum_i q_i p_i (2^{n-1} -1) /2\nonumber\\
    &= (2^{n-1} - 1) /2,
\end{align}
where in the second line we employ the convexity of negativity.
In the third line we apply local untary operations on the party $C$ to rotate the states $|\psi_i\rangle$'s to $\ket{0}$.
In the fourth line, we use fact that negativity is invariant under local unitaries and adding local ancillas, see \cite{horodecki2005simplifying}.
{In the fifth line, we apply the upper bound given in Eq.~(\ref{eq:uppernegativity}).}

On the other hand, we obtain
\begin{align}
    N_{BC|A}[\Phi_C(\beta_{ABC})]
    &\ge N_{BC|A}[\tr_C(\Phi_C(\beta_{ABC}))\otimes |0\rangle\langle 0|_C]\nonumber\\
    &= N_{B|A}[\tr_C(\beta_{ABC})]\nonumber\\
    &= N_{B|A}[(|\Psi^+\rangle\langle \Psi^+|_{AB})^{\otimes (n-1)}\otimes \tr_C(|\Psi^+\rangle\langle \Psi^+|_{AC})]\nonumber\\
    &= N_{B|A}[(|\Psi^+\rangle\langle \Psi^+|_{AB})^{\otimes (n-1)}]\nonumber\\
    &= (2^{n-1} - 1) /2.
\end{align}
In the first line we use the LOCC monotonicity, and in the second line we make use of the fact that $\tr_C \circ \, \Phi_C = \tr_C$.
In the fourth line, we use fact that negativity is invariant under adding local ancillas, see \cite{horodecki2005simplifying}.

Thus, independently of the entanglement breaking channel $\Phi_C$, we show
\begin{align}
    N_{BC|A}[\Phi_C(\beta_{ABC})] = (2^{n-1} - 1) /2.
\end{align}
This result directly leads to
\begin{equation}\label{nega:bc|a}
    {N}_{BC|A}(\beta_{ABC}) - {N}_{BC|A}\left[\Phi_C(\beta_{ABC})\right] = 2^{n-2}.
\end{equation}
Also, since negativity is invariant under adding local ancillas, we have
\begin{align}\label{eq:g10}
    {N}_{B|CA}(\beta_{ABC}) = {N}_{B|CA}\left[\Phi_C(\beta_{ABC})\right] = {N}_{B|A}\left[(|\Psi^+\rangle\langle \Psi^+|_{AB})^{\otimes (n-1)}\right] = (2^{n-1}-1)/2,
\end{align}
which implies
\begin{align}\label{nega:b|ca}
    {N}_{B|CA}(\beta_{ABC}) - {N}_{B|CA}\left[\Phi_C(\beta_{ABC})\right] = 0.
\end{align}
Similarly, we have
\begin{align}\label{eq:g12}
{N}_{AB|C}(\beta_{ABC}) = {N}_{A|C}\left[|\Psi^+\rangle\langle \Psi^+|_{AC}\right] = 1/2.
\end{align}
The fact that $\Phi_C$ is an entanglement-breaking channel implies that 
\begin{equation}
    {N}_{AB|C}\left[\Phi_C(\beta_{ABC})\right] = 0.
\end{equation}
Consequently, we have
\begin{align}\label{nega:ab|c}
 {N}_{AB|C}(\beta_{ABC}) - {N}_{AB|C}\left[\Phi_C(\beta_{ABC})\right] = 1/2.
\end{align}

By definition of $\Delta_{\mathcal{E}}(\beta_{ABC})$ with $N_{ABC}$ using Eqs.~(\ref{nega:bc|a}, \ref{nega:b|ca}, \ref{nega:ab|c}),  we complete the proof:
\begin{equation}
\Delta_{\mathcal{E}}(\beta_{ABC}) = 2^{n-2} + 1/2.
\end{equation}
{
We have one remark. From Eq.~\eqref{eq:g4}, Eq.~\eqref{eq:g10} and Eq.~\eqref{eq:g12}, we know that the original tripartite negativity is
\begin{equation}
    N_{ABC}(\beta_{ABC}) = 2^{n-1} + 2^{n-2} - 1/2,
\end{equation}
which is strictly larger than $\Delta_{\mathcal{E}}(\beta_{ABC})$ whenever $n\ge 2$. Furthermore, $N_{ABC}(\beta_{ABC})/\Delta_{\mathcal{E}}(\beta_{ABC})$ goes to $2$ as $n$ goes to infinity.
}
\end{proof}


\section{Observations on complete entanglement loss under classicalization}\label{ap:h}

In this Appendix, we propose two observations for the entangled states satisfying Condition~7. A similar observation has been made for pure states in Ref.~\cite{neven2018entanglement}.
\addtocounter{theorem}{1}
\begin{observation} \label{ob:subrank2}
Suppose that a tripartite state $\rho_{ABC}$ satisfies Condition 7.
If $\rho_{ABC}$ is entangled {for the bipartitions $A|BC$ and $B|AC$}, then the reduced state  $\rho_{AB} = \tr_C(\rho_{ABC})$ should have rank more than $2$.
\end{observation}

We remark that the generalization of Observation~\ref{ob:subrank2} to the $n$-partite case is given in Appendix~\ref{ap:j} for $n>3$.

\begin{proof}
First we denote that
$p_x = \tr[\langle x|\rho_{ABC} |x\rangle]$
and {$\sigma_{\ket{x}} = \langle x|\rho_{ABC} |x\rangle/p_x$}.
Let us begin by recalling that any tripartite quantum state can be written as
\begin{align}
    \rho_{ABC} = \sum_{i,j} M_{ij} \otimes \ket{i}\!\bra{j},
\end{align}
where
$M_{ij}=\tr_C[\rho_{ABC} ( \id_{AB} \otimes \ket{j}\!\bra{i})]$.
For $i=j$, we have that $M_{ii} = p_i \sigma_{\ket{i}}$.
For $i\neq j$, $M_{ij}$ can be written as linear combinations of $p_x \sigma_{\ket{x}}$ for some $\ket{x}$, since any $\ket{j}\!\bra{i}$ can be decomposed using some projectors $\ket{x}\!\bra{x}$.
The more explicit form will be given below.

In the following, we will show the {contraposition} of the observation, that is, if $\rho_{ABC}$ satisfies Condition 7 and $\rho_{AB}$ has rank no more than $2$, then $\rho_{ABC}$ is {either separable for the bipartition $A|BC$ or separable for the bipartition $B|AC$}.
{
Since $\rho_{AB} = \sum_i p_i \sigma_{\ket{i}}$ where $\{\ket{i}\}$ is the computational orthonormal basis, and $\sigma_{\ket{i}}$ is separable for any $\ket{i}$ according to Condition 7, then $\rho_{AB}$ is also separable.
}
If $\rho_{AB}$ {has} rank $1$, it is easy to see that $\rho_{ABC}$ is a pure product state.
{
Further, let us consider the case that the separable state $\rho_{AB}$ has exactly rank $2$. Up to local unitary, we can assume the following decomposition:
\begin{align}
    \rho_{AB} = \alpha(\lambda \ket{00}\!\bra{00}
    +(1-\lambda)\ket{ab}\!\bra{ab}) + (1-\alpha) \sum_{i} \lambda_i |a_ib_i\rangle\langle a_ib_i|,
\end{align}
where $\ket{ab} \neq \ket{00}$, $\alpha, \lambda, \lambda_i \in [0,1]$.  

Denote $|\psi_1\rangle, |\psi_2\rangle$ the eigenstates of $\rho_{AB}$ with non-zero eigenvalues. Then $\ket{00}, \ket{ab}, \ket{a_ib_i}$ should be superpositions of $|\psi_1\rangle, |\psi_2\rangle$. Since $\ket{ab} \neq \ket{00}$, $|\psi_1\rangle, |\psi_2\rangle$ can also be written as superpositions of $\ket{00}, \ket{ab}$. Consequently, any $\ket{a_ib_i}$ can be written as superpositions of $\ket{00}, \ket{ab}$.

In the case that $|a\rangle = |0\rangle$, we have $|a_i\rangle = |0\rangle$, which implies that $\rho_A = \tr_{BC}(\rho_{ABC}) = \tr_B(\rho_{AB}) = |0\rangle\langle 0|$. Hence, $\rho_{ABC} = |0\rangle\langle 0|\otimes \rho_{BC}$, which contradicts the assumption that $\rho_{ABC}$ is entangled for the bipartition $A|BC$. Thus, $|a\rangle \neq |0\rangle$ should hold. Similarly, we have $|b\rangle \neq |0\rangle$.

Since $|a\rangle \neq |0\rangle$, $|b\rangle \neq |0\rangle$, then any non-trivial superposition of them is entangled. This leads to that $\ket{a_ib_i}$ should either be $|00\rangle$ or $|ab\rangle$ up to a phase.
}

Since the range of $\sigma_{\ket{x}}$ belongs to the range of $\rho_{AB}$ {and $\sigma_{\ket{x}}$ is separable}, we have
\begin{align}
    \sigma_{\ket{x}} =
    \lambda_x \ket{00}\!\bra{00}
    +(1-\lambda_x)\ket{ab}\!\bra{ab},
\end{align}
where $\sum_x p_x \lambda_x =\lambda$.
{Since $M_{ij}$ is a combination of $\sigma_{|x\rangle}$, $M_{ij}$ can be written as}
\begin{align}
    M_{ij} = X_{ij} \ket{00}\!\bra{00}
    +Y_{ij}\ket{ab}\!\bra{ab},
\end{align}
where
the coefficients $X_{ij}$ and $Y_{ij}$ are given by
combinations of $p_x \lambda_x$ for some $x$.
Accordingly, we can write
\begin{align}
    \rho_{ABC}
    =  \ket{00}\!\bra{00} \otimes \tau_x
    + \ket{ab}\!\bra{ab} \otimes \tau_y,
\end{align}
where
$\tau_x=\sum_{i,j} X_{ij} \ket{i}\!\bra{j}$ and 
$\tau_y=\sum_{i,j} Y_{ij} \ket{i}\!\bra{j}$.

To show that $\rho_{ABC}$ is fully separable, it is sufficient to prove that
the matrices $\tau_x$ and $\tau_y$ are positive semidefinite.
For that, we note that since $\ket{ab} \neq \ket{00}$, there exists a bipartite pure state $\ket{\alpha \beta}$ such that
$\braket{ab|\alpha \beta} = 0$ and $\braket{00|\alpha \beta} \neq 0$.
Then it holds that
\begin{equation}
    \braket{\alpha \beta \gamma|\rho_{ABC}|\alpha \beta \gamma}
    = |\braket{\alpha \beta|00}|^2 \braket{\gamma|\tau_x|\gamma} \geq 0,
\end{equation}
for any $\ket{\gamma}$.
This implies that $\braket{\gamma|\tau_x|\gamma} \geq 0$, that is, $\tau_x$ is positive semidefinite.
Similarly, we can show that $\tau_y$ is positive semidefinite.
Hence, we conclude that $\rho_{ABC}$ is fully separable{, which contradicts the assumption}.
\end{proof}

In the case that the party $C$ is not entangled with $A$ and $B$, we have a similar requirement of the global state as in the following observation.
\begin{observation} \label{ob:wholerank2}
Suppose that a tripartite state $\rho_{ABC}$ satisfies Condition 7.
If $\rho_{ABC}$ is {entangled {for the bipartitions $A|BC$ and $B|AC$}} separable for $AB|C$, then it should have rank more than $2$.
\end{observation}
\begin{proof}
Here we prove {the statement} by contradiction.
Let us assume $\rho_{ABC}$ satisfies Condition 7 and has rank no more than $2$. Since $\rho_{ABC}$ is separable for the bipartition $AB|C$, we have the decomposition
\begin{equation}
     \rho_{ABC} = \sum_{i} p_i \ket{\psi_i\phi_i}\bra{\psi_i\phi_i},
  \end{equation}
  where $\ket{\psi_i}, \ket{\phi_i}$ are states for parties $A, B$ and party $C$, respectively.
  
 By assumption, the dimension of the space spanned by $\{\ket{\psi_i\phi_i}\}$ is no more than $2$, this leads to that the dimension of the space spanned by $\{\ket{\psi_i}\}$ is no more than $2$. Thus, $\rho_{AB} = \tr_C(\rho_{ABC}) = \sum_i p_i \ket{\psi_i}\!\bra{\psi_i}$ has rank no more than $2$. By applying Observation~\ref{ob:subrank2}, we finish the proof.
\end{proof}
{
We have two remarks.
First, one can indeed find tripartite entangled states satisfying Condition 7 and separable for the bipartition $AB|C$.
{Especially, there exist tripartite entangled states which are separable for any bipartition~\cite{bennett1999unextendible,acin2001classification}, which satisfy Condition 7 automatically.}
We collect more such examples in Appendix~\ref{ap:i}.
Second, Observations~\ref{ob:subrank2}, \ref{ob:wholerank2} may provide insight into a type of quantum marginal problem: whether a global state can be separable or entangled if its marginal systems are subjected to separability conditions and rank constraints.
}

\section{Examples for three-qubit states}\label{ap:i}
Here, we discuss three-qubit entangled states that satisfy Condition 7 for the complete entanglement change.
In this Appendix, we will first propose a nontrivial three-qubit state that is entangled $A|BC$ and $AC|B$ but separable for $AB|C$.
Next, we will connect the complete entanglement change with bound entanglement.

\subsection*{I.1: Complete entanglement change with separability for $AB|C$}
To find a nontrivial three-qubit entangled state that satisfy Condition 7, we employ the method of entanglement witnesses:
For an Hermitian operator $W$, it is called an entanglement witness if
$\mathrm{tr} (W \rho_{s}) \geq 0$ for all separable states $\rho_{s}$,
and 
$\mathrm{tr} (W \rho_{e}) < 0$ for some entangled states $\rho_{e}$.
The latter allows us to detect entanglement.
In particular, we adopt the entanglement witness that can have the negative eigenvalues of its partial transpose (NPT) state.
This witness is described as follows:
Suppose that a state $\rho_{e}$ is NPT.
Then one can find a negative eigenvalue $\lambda <0$ of $\rho_{e}^{T_A}$ and the corresponding eigenvector $\ket{\phi_C}$.
Hence the operator $\ket{\phi_C}\!\bra{\phi_C}^{T_A}$ can be an witness to detect the entangled state $\rho_{e}$.

In practice, entanglement witnesses can be implemented by semi-definite programming (SDP).
For our purpose, we use the following conditions that are compatible with the SDP method.
First, to impose the separability condition for the bipartition $AB|C$, we apply the fact that if a $2 \otimes N$ state $\rho_{XY}$ obeys $\rho_{XY} = \rho_{XY}^{T_X}$, then it is separable, see Theorem $2$ in Ref.~\cite{kraus2000separability}.
That is, we require that $\rho_{ABC}=\rho_{ABC}^{T_C}$.
Second, for the separability condition of the two-qubit post-measurement state $\sigma_{\ket{x}}$, we employ the positive partial transpose (PPT) criterion, which is necessary and sufficient for two-qubit separability.
Third, for the sake of simplicity, we suppose that the state $\rho_{ABC}$ is invariant under exchange between $A$ and $B$ using SWAP operator $\mathrm{SWAP} \ket{a}\ket{b} = \ket{b}\ket{a}$.

Since the set of NPT states is not convex, we use the see-saw method with entanglement witnesses.
This is a numerical iteration technique for non-convex optimization, which allows us to find states with the (local) minimal value as a solution.
From the numerical solution, we can find an analytical form of the state and verify that it satisfies Condition $1$ for any measurement direction.
Our finding is the following entangled state:
\begin{align}
  \Tilde{\rho} &=\frac{1}{8}
\begin{bmatrix}
     0 & 0 & 0 & 0 & 0 & 0 & 0 & 0 \\
     0 & 2 & 0 & 0 & 0 & 0 & 0 & 0 \\
     0 & 0 & 1 & 0 & 0 & 1 & 0 & 0 \\
     0 & 0 & 0 & 1 & 1 & 0 & 0 & 0 \\
     0 & 0 & 0 & 1 & 1 & 0 & 0 & 0 \\
     0 & 0 & 1 & 0 & 0 & 1 & 0 & 0 \\
     0 & 0 & 0 & 0 & 0 & 0 & 2 & 0 \\
     0 & 0 & 0 & 0 & 0 & 0 & 0 & 0 \\
\end{bmatrix}.
\end{align}
This state has the following properties.
First, the matrix rank of $\Tilde{\rho}$ is $4$.
Second, one can show that the state $\sigma_{\ket{x}}$ with $\ket{x} = (\cos t, e^{i a} \sin t)$ is PPT and therefore separable for any $t,a$.
Third, the minimum eigenvalue of $\Tilde{\rho}^{T_A}$ is equal to $-1/8$.
Fourth, the party $C$ is not entangled with the other two parties. Nevertheless, the discord $D_{AB|C}(\tilde{\rho})>0$, which is necessary for complete entanglement change according to Observation~\ref{ob:relative2}.

\subsection*{I.2: Complete entanglement change and bound entanglement}
We have found the existence of state $\Tilde{\rho}$ that is entangled states for $A|BC$ and $AC|B$ but separable for $AB|C$ that can achieve the complete entanglement change.
Now we are also interested in the case where the separability for $AB|C$ is replaced by bound entanglement.
Such a state is already known as the $4\otimes 2$ bound entangled state~\cite{horodecki1997separability}, denoted by
\begin{equation}
\rho_{\rm HDK}=\frac{1}{h}
\begin{bmatrix}
 2 t & 0 & 0 & 0 & 0 & 0 & 2 t & 0 \\
 0 & 2 t & 0 & 0 & 0 & 0 & 0 & 2 t \\
 0 & 0 & t+1 & 0 & 0 & 0 & 0 & t^\prime \\
 0 & 0 & 0 & 2 t & 2 t & 0 & 0 & 0 \\
 0 & 0 & 0 & 2 t & 2 t & 0 & 0 & 0 \\
 0 & 0 & 0 & 0 & 0 & 2 t & 0 & 0 \\
 2 t & 0 & 0 & 0 & 0 & 0 & 2 t & 0 \\
 0 & 2 t & t^\prime & 0 & 0 & 0 & 0 & t+1
\end{bmatrix},
\end{equation}
where $t^\prime=\sqrt{1-t^2}$, $h={2(1+7t)}$ and $0<t<1$.
Here the parties $AB$ are in $4$-dimensional systems and the party $C$ is $2$-dimensional systems.
We remark that this state satisfies Condition $1$.
Since this state is NPT entangled for $A|BC$ and $AC|B$ but PPT entangled for $AB|C$, we cannot apply Observation~\ref{ob:wholerank2}.
On the other hand, its reduced state $\rho_{AB}$ has rank $4$, and therefore, it complies with Observation~\ref{ob:subrank2}.

To proceed further, we now present the following:
\begin{observation}
  If a tripartite state $\rho_{ABC}$ is separable either for the bipartition $A|BC$ or the bipartition $B|AC$, then $\rho_{ABC}$ satisfies Condition $1$.
\end{observation}
\begin{proof}
  If $\rho_{ABC}$ is separable either for $A|BC$ or $B|AC$, then the normalized state of $\bra{x}\rho_{ABC}\ket{x}$ is separable for any measurement direction $\ket{x}$ on $C$.
  Thus, Observation~\ref{ob:lower2} implies that the  entanglement change must be complete.
\end{proof}
In the following, we collect entangled states for complete entanglement change which are even separable for any bipartition:
\begin{align}
    \rho_{\rm UPB} &=\frac{1}{32}
\begin{bmatrix}
 7 & 1 & 1 & \bar{1} & 1 & \bar{1} & \bar{1} & 1 \\
 1 & 3 & \bar{1} & 1 & \bar{1} & \bar{3} & 1 & \bar{1} \\
 1 & \bar{1} & 3 & \bar{3} & \bar{1} & 1 & 1 & \bar{1} \\
 \bar{1} & 1 & \bar{3} & 3 & 1 & \bar{1} & \bar{1} & 1 \\
 1 & \bar{1} & \bar{1} & 1 & 3 & 1 & \bar{3} & \bar{1} \\
 \bar{1} & \bar{3} & 1 & \bar{1} & 1 & 3 & \bar{1} & 1 \\
 \bar{1} & 1 & 1 & \bar{1} & \bar{3} & \bar{1} & 3 & 1 \\
 1 & \bar{1} & \bar{1} & 1 & \bar{1} & 1 & 1 & 7 \\
\end{bmatrix},
\,\,\,\,\,\,\,\,\,\,\,\,\,\,\,\,\,\,\,\,\,\,
\rho_{\rm ADMA} =\frac{1}{n}
    \begin{bmatrix}
     1 & 0 & 0 & 0 & 0 & 0 & 0 & 1 \\
     0 & a & 0 & 0 & 0 & 0 & 0 & 0 \\
     0 & 0 & b & 0 & 0 & 0 & 0 & 0 \\
     0 & 0 & 0 & c & 0 & 0 & 0 & 0 \\
     0 & 0 & 0 & 0 & \frac{1}{c} & 0 & 0 & 0 \\
     0 & 0 & 0 & 0 & 0 & \frac{1}{b} & 0 & 0 \\
     0 & 0 & 0 & 0 & 0 & 0 & \frac{1}{a} & 0 \\
     1 & 0 & 0 & 0 & 0 & 0 & 0 & 1 \\
    \end{bmatrix},\\
\rho_{AK} &=\frac{1}{8(1+y)}
\begin{bmatrix}
 x & 0 & 0 & 0 & 0 & 0 & 0 & 2 \\
 0 & y & 0 & 0 & 0 & 0 & 2 & 0 \\
 0 & 0 & y & 0 & 0 & \bar{2} & 0 & 0 \\
 0 & 0 & 0 & y & 2 & 0 & 0 & 0 \\
 0 & 0 & 0 & 2 & y & 0 & 0 & 0 \\
 0 & 0 & \bar{2} & 0 & 0 & y & 0 & 0 \\
 0 & 2 & 0 & 0 & 0 & 0 & y & 0 \\
 2 & 0 & 0 & 0 & 0 & 0 & 0 & x \\
\end{bmatrix},
\,\,\,\,\,\,\,\,
\rho_{PH} =\frac{1}{m} 
\begin{bmatrix}
 2 z & 0 & 0 & 0 & 0 & 0 & 0 & 0 \\
 0 & 1 & 1 & 0 & 1 & 0 & 0 & 0 \\
 0 & 1 & 1 & 0 & 1 & 0 & 0 & 0 \\
 0 & 0 & 0 & \frac{1}{z} & 0 & 0 & 0 & 0 \\
 0 & 1 & 1 & 0 & 1 & 0 & 0 & 0 \\
 0 & 0 & 0 & 0 & 0 & \frac{1}{z} & 0 & 0 \\
 0 & 0 & 0 & 0 & 0 & 0 & \frac{1}{z} & 0 \\
 0 & 0 & 0 & 0 & 0 & 0 & 0 & 0 \\
\end{bmatrix},
\end{align}
where $\bar{1} = -1, \bar{2} = -2,  \bar{3} = -3$, $a,b,c,x,y,z > 0$,
$abc\neq 1$, $x=y+4$, $n=2 + 1/a + a + 1/b + b + 1/c + c$, and $m=3 + 3/z + 2 z$.
These states have been already known:
$\rho_{\rm UPB}$ in Ref.~\cite{bennett1999unextendible},
$\rho_{\rm ADMA}$ in Ref.~\cite{acin2001classification},
$\rho_{AK}$ in Ref.~\cite{kay2011optimal},
and {the Hyllus state}
$\rho_{PH}$ in Eq.~(2.105) in Ref.~\cite{hyllus2005witnessing}.
Note that $\rho_{AK}$ is entangled for $2\leq y\leq 2.828$ but separable for $y \geq 2\sqrt{2}$.
Also $\rho_{\rm UPB}$ is permutationally symmetric.

Let us summarize the property of these states.
The first common property of them is that they are separable for any bipartition, but not fully separable.
In that sense, they are not multipartite distillable and then bound entangled~\cite{guhne2011entanglement}.
Here we remark that GHZ diagonal states that are PPT for any bipartition are separable for any bipartition \cite{nagata2009necessary}.
Second, their matrix ranks are, respectively, given by
$\text{Rank}(\rho_{\rm UPB})=4,
\,
\text{Rank}(\rho_{\rm ADMA})=7,
\,
\text{Rank}(\rho_{\rm AK})=8,
\,
\text{Rank}(\rho_{\rm PH})=5.$
This follows the results of Observation~\ref{ob:wholerank2}.
Finally, these bound entangled states can be detected with the help of the previously presented entanglement criteria in Refs.~\cite{horodecki2006separability, guhne2010separability, guhne2011entanglement}.

The last example is the three-qubit thermal state with Heisenberg chain model:
\begin{align}
    &\rho_H = \exp{(-H_H/T)}/Z,\\
    &H_{H}=\sum_{i=1,2,3}
    \sigma_X^i \sigma_X^{i+1}
    +\sigma_Y^i \sigma_Y^{i+1}
    +\sigma_Z^i \sigma_Z^{i+1},
\end{align}
with temperature $T$ and $Z=\tr[\exp{(-H_H/T)}]$.
This thermal state has been shown to be bound entangled in the temperature range $T \in [4.33, 5.46]$, in the sense that they are separable for any bipartition but not fully separable in Refs.~\cite{eggeling2001separability, toth2007optimal} and Table II in \cite{toth2009spin}, where the bound entanglement can be detected by the optimal spin squeezing inequality.

\section{Generalization of Observation~\ref{ob:subrank2}}\label{ap:j}
\begin{observation}\label{ob:multisubrank}
Let $\rho_{A_1\ldots A_{n-1} A_{n}}$ be a $n$-partite quantum state
and let $P_{n}^{\ket{x}}=\ket{x}\bra{x}$ be a projector on
the subsystem $A_{n}$ with $\sum_x P_{n}^{\ket{x}}=I$.
Suppose that the normalized state
$\sigma^{\ket{x}} = \tr_{n}(P_{n}^{\ket{x}} \rho_{A_1\ldots  A_{n-1} A_{n}})/p_n$ with $p_n=\tr(P_{n}^{\ket{x}} \rho_{A_1\ldots  A_{n-1} A_{n}})$
is fully separable for any $\ket{x}$,
and the reduced state
$\rho_{A_1\ldots A_{n-1}}=\tr_{n} (\rho_{A_1\ldots A_{n-1} A_{n}})$ can be written as
\begin{equation}
    \rho_{A_1\ldots A_{n-1}} = \sum_{i=1}^k p_i |\psi_i\rangle\langle \psi_i|,
\end{equation}
where $\{|\psi_i\rangle\}_{i=1}^k$ are linearly
independent fully product states, i.e.,
$\ket{\psi_i}=\bigotimes_{j=1}^n \ket{\psi_j^i}$, and any superposition of $\{|\psi_i\rangle\}_{i=1}^k$ given by
$\sum_i c_i \ket{\psi_i}$ is not a fully product state.
In this case, $\rho_{A_1\ldots A_{n-1} A_{n}}$ should be fully separable.
\end{observation}
\begin{proof}
We begin by recalling that any $n$-particle state can be written as
\begin{align}
    \rho_{A_1\ldots A_{n-1} A_{n}} = \sum_{i,j} M_{ij}\otimes \ket{i}\!\bra{j},
\end{align}
where $M_{ij}$ is a matrix on the $A_1\ldots A_{n-1}$ system and $\ket{i}\!\bra{j}$ is on the $A_n$ system.
Then, from the assumption, we notice
\begin{align}
  \rho_{A_1\ldots A_{n-1}}=\sum_i M_{ii}=\sum_{i=1}^k p_i |\psi_i\rangle\langle \psi_i|.
\end{align}
Note that $M_{ii}=\sigma^{\ket{i}}$,
which implies that the range of $\sigma^{\ket{x}}$ is in the subspace
spanned by $\{\ket{\psi_i}\}$. From the assumption that $\sigma^{\ket{x}}$ is separable and any superposition of $\{\ket{\psi_i}\} $  is entangled, we have
\begin{align}
     \sigma^{\ket{x}} = \sum_j q_j^x \ket{\psi_j}\!\bra{\psi_j}.
\end{align}
Also $M_{ij}$ can be written in the linear combination
of $\sigma^{\ket{x}}$.
Accordingly, we have
\begin{align}
    \rho_{A_1\ldots A_{n-1} A_{n}}
    = \sum_{i,j,k} c_{ijk} \ket{\psi_k}\!\bra{\psi_k}
    \otimes \ket{i}\!\bra{j}
    = \sum_{k} \ket{\psi_k}\!\bra{\psi_k}
    \otimes \tau_k,
\end{align}
where 
$\tau_k =\sum_{ij} c_{ijk} \ket{i}\!\bra{j}$, and $c_{ijk}$ is the coefficient of $\ket{\psi_k}\!\bra{\psi_k}$ when we expand $M_{ij}$.
Below we show that $\tau_k$ is positive semidefinite.

From the assumption that $\{|\psi_i\rangle\}_{i=1}^k$
are linearly independent, we know that there are states
$\{|\phi_i\rangle\}_{i=1}^k$ such that
\begin{align}
  &\braket{\psi_i|\phi_j} = 0,\text{ if } i\neq j,\ \braket{\psi_i|\phi_i} > 0.
\end{align}
Then, for any $\ket{v}$,
\begin{equation}
    \bra{\phi_i v} \rho_{A_1\ldots A_{n-1} A_{n}}
    \ket{\phi_i v} = \braket{\phi_i|\psi_i}^2 \bra{v}\tau_i\ket{v} \ge 0,
\end{equation}
that is, $\bra{v}\tau_i\ket{v} \ge 0$. This implies that $\tau_i$ is positive semidefinite. Hence, $\rho_{A_1\ldots A_{n-1} A_{n}}$
is a fully separable state.
\end{proof}
\bibliography{references.bib}

\begin{thebibliography}{53}%
\makeatletter
\providecommand \@ifxundefined [1]{%
 \@ifx{#1\undefined}
}%
\providecommand \@ifnum [1]{%
 \ifnum #1\expandafter \@firstoftwo
 \else \expandafter \@secondoftwo
 \fi
}%
\providecommand \@ifx [1]{%
 \ifx #1\expandafter \@firstoftwo
 \else \expandafter \@secondoftwo
 \fi
}%
\providecommand \natexlab [1]{#1}%
\providecommand \enquote  [1]{``#1''}%
\providecommand \bibnamefont  [1]{#1}%
\providecommand \bibfnamefont [1]{#1}%
\providecommand \citenamefont [1]{#1}%
\providecommand \href@noop [0]{\@secondoftwo}%
\providecommand \href [0]{\begingroup \@sanitize@url \@href}%
\providecommand \@href[1]{\@@startlink{#1}\@@href}%
\providecommand \@@href[1]{\endgroup#1\@@endlink}%
\providecommand \@sanitize@url [0]{\catcode `\\12\catcode `\$12\catcode
  `\&12\catcode `\#12\catcode `\^12\catcode `\_12\catcode `\%12\relax}%
\providecommand \@@startlink[1]{}%
\providecommand \@@endlink[0]{}%
\providecommand \url  [0]{\begingroup\@sanitize@url \@url }%
\providecommand \@url [1]{\endgroup\@href {#1}{\urlprefix }}%
\providecommand \urlprefix  [0]{URL }%
\providecommand \Eprint [0]{\href }%
\providecommand \doibase [0]{https://doi.org/}%
\providecommand \selectlanguage [0]{\@gobble}%
\providecommand \bibinfo  [0]{\@secondoftwo}%
\providecommand \bibfield  [0]{\@secondoftwo}%
\providecommand \translation [1]{[#1]}%
\providecommand \BibitemOpen [0]{}%
\providecommand \bibitemStop [0]{}%
\providecommand \bibitemNoStop [0]{.\EOS\space}%
\providecommand \EOS [0]{\spacefactor3000\relax}%
\providecommand \BibitemShut  [1]{\csname bibitem#1\endcsname}%
\let\auto@bib@innerbib\@empty
\bibitem [{\citenamefont {Chitambar}\ and\ \citenamefont
  {Gour}(2019)}]{chitambar2019quantum}%
  \BibitemOpen
  \bibfield  {author} {\bibinfo {author} {\bibfnamefont {E.}~\bibnamefont
  {Chitambar}}\ and\ \bibinfo {author} {\bibfnamefont {G.}~\bibnamefont
  {Gour}},\ }\bibfield  {title} {\bibinfo {title} {Quantum resource theories},\
  }\href {https://doi.org/https://doi.org/10.1103/RevModPhys.91.025001}
  {\bibfield  {journal} {\bibinfo  {journal} {Rev. Mod. Phys.}\ }\textbf
  {\bibinfo {volume} {91}},\ \bibinfo {pages} {025001} (\bibinfo {year}
  {2019})}\BibitemShut {NoStop}%
\bibitem [{\citenamefont {DiVincenzo}(1995)}]{divincenzo1995quantum}%
  \BibitemOpen
  \bibfield  {author} {\bibinfo {author} {\bibfnamefont {D.~P.}\ \bibnamefont
  {DiVincenzo}},\ }\bibfield  {title} {\bibinfo {title} {Quantum computation},\
  }\href {https://doi.org/https://doi.org/10.1126/science.270.5234.255}
  {\bibfield  {journal} {\bibinfo  {journal} {Science}\ }\textbf {\bibinfo
  {volume} {270}},\ \bibinfo {pages} {255} (\bibinfo {year}
  {1995})}\BibitemShut {NoStop}%
\bibitem [{\citenamefont {Scarani}\ \emph {et~al.}(2009)\citenamefont
  {Scarani}, \citenamefont {Bechmann-Pasquinucci}, \citenamefont {Cerf},
  \citenamefont {Du{\v{s}}ek}, \citenamefont {L{\"u}tkenhaus},\ and\
  \citenamefont {Peev}}]{scarani2009security}%
  \BibitemOpen
  \bibfield  {author} {\bibinfo {author} {\bibfnamefont {V.}~\bibnamefont
  {Scarani}}, \bibinfo {author} {\bibfnamefont {H.}~\bibnamefont
  {Bechmann-Pasquinucci}}, \bibinfo {author} {\bibfnamefont {N.~J.}\
  \bibnamefont {Cerf}}, \bibinfo {author} {\bibfnamefont {M.}~\bibnamefont
  {Du{\v{s}}ek}}, \bibinfo {author} {\bibfnamefont {N.}~\bibnamefont
  {L{\"u}tkenhaus}},\ and\ \bibinfo {author} {\bibfnamefont {M.}~\bibnamefont
  {Peev}},\ }\bibfield  {title} {\bibinfo {title} {The security of practical
  quantum key distribution},\ }\href
  {https://doi.org/https://doi.org/10.1103/RevModPhys.81.1301} {\bibfield
  {journal} {\bibinfo  {journal} {Rev. Mod. Phys.}\ }\textbf {\bibinfo {volume}
  {81}},\ \bibinfo {pages} {1301} (\bibinfo {year} {2009})}\BibitemShut
  {NoStop}%
\bibitem [{\citenamefont {Giovannetti}\ \emph {et~al.}(2006)\citenamefont
  {Giovannetti}, \citenamefont {Lloyd},\ and\ \citenamefont
  {Maccone}}]{giovannetti2006quantum}%
  \BibitemOpen
  \bibfield  {author} {\bibinfo {author} {\bibfnamefont {V.}~\bibnamefont
  {Giovannetti}}, \bibinfo {author} {\bibfnamefont {S.}~\bibnamefont {Lloyd}},\
  and\ \bibinfo {author} {\bibfnamefont {L.}~\bibnamefont {Maccone}},\
  }\bibfield  {title} {\bibinfo {title} {Quantum metrology},\ }\href
  {https://doi.org/https://doi.org/10.1103/PhysRevLett.96.010401} {\bibfield
  {journal} {\bibinfo  {journal} {Phys. Rev. Lett.}\ }\textbf {\bibinfo
  {volume} {96}},\ \bibinfo {pages} {010401} (\bibinfo {year}
  {2006})}\BibitemShut {NoStop}%
\bibitem [{\citenamefont {Lidar}\ and\ \citenamefont
  {Brun}(2013)}]{lidar2013quantum}%
  \BibitemOpen
  \bibfield  {author} {\bibinfo {author} {\bibfnamefont {D.~A.}\ \bibnamefont
  {Lidar}}\ and\ \bibinfo {author} {\bibfnamefont {T.~A.}\ \bibnamefont
  {Brun}},\ }\href {https://doi.org/https://doi.org/10.1017/CBO9781139034807}
  {\emph {\bibinfo {title} {Quantum error correction}}}\ (\bibinfo  {publisher}
  {Cambridge University Press},\ \bibinfo {year} {2013})\BibitemShut {NoStop}%
\bibitem [{\citenamefont {Horodecki}\ \emph {et~al.}(2009)\citenamefont
  {Horodecki}, \citenamefont {Horodecki}, \citenamefont {Horodecki},\ and\
  \citenamefont {Horodecki}}]{horodecki2009quantum}%
  \BibitemOpen
  \bibfield  {author} {\bibinfo {author} {\bibfnamefont {R.}~\bibnamefont
  {Horodecki}}, \bibinfo {author} {\bibfnamefont {P.}~\bibnamefont
  {Horodecki}}, \bibinfo {author} {\bibfnamefont {M.}~\bibnamefont
  {Horodecki}},\ and\ \bibinfo {author} {\bibfnamefont {K.}~\bibnamefont
  {Horodecki}},\ }\bibfield  {title} {\bibinfo {title} {Quantum entanglement},\
  }\href {https://doi.org/https://doi.org/10.1103/RevModPhys.81.865} {\bibfield
   {journal} {\bibinfo  {journal} {Rev. Mod. Phys.}\ }\textbf {\bibinfo
  {volume} {81}},\ \bibinfo {pages} {865} (\bibinfo {year} {2009})}\BibitemShut
  {NoStop}%
\bibitem [{\citenamefont {G{\"u}hne}\ and\ \citenamefont
  {T{\'o}th}(2009)}]{guhne2009entanglement}%
  \BibitemOpen
  \bibfield  {author} {\bibinfo {author} {\bibfnamefont {O.}~\bibnamefont
  {G{\"u}hne}}\ and\ \bibinfo {author} {\bibfnamefont {G.}~\bibnamefont
  {T{\'o}th}},\ }\bibfield  {title} {\bibinfo {title} {Entanglement
  detection},\ }\href
  {https://doi.org/https://doi.org/10.1016/j.physrep.2009.02.004} {\bibfield
  {journal} {\bibinfo  {journal} {Phys. Rep.}\ }\textbf {\bibinfo {volume}
  {474}},\ \bibinfo {pages} {1} (\bibinfo {year} {2009})}\BibitemShut {NoStop}%
\bibitem [{\citenamefont {Simon}\ and\ \citenamefont
  {Kempe}(2002)}]{simon2002robustness}%
  \BibitemOpen
  \bibfield  {author} {\bibinfo {author} {\bibfnamefont {C.}~\bibnamefont
  {Simon}}\ and\ \bibinfo {author} {\bibfnamefont {J.}~\bibnamefont {Kempe}},\
  }\bibfield  {title} {\bibinfo {title} {Robustness of multiparty
  entanglement},\ }\href
  {https://doi.org/https://doi.org/10.1103/PhysRevA.65.052327} {\bibfield
  {journal} {\bibinfo  {journal} {Phys. Rev. A}\ }\textbf {\bibinfo {volume}
  {65}},\ \bibinfo {pages} {052327} (\bibinfo {year} {2002})}\BibitemShut
  {NoStop}%
\bibitem [{\citenamefont {D{\"u}r}\ and\ \citenamefont
  {Briegel}(2004)}]{dur2004stability}%
  \BibitemOpen
  \bibfield  {author} {\bibinfo {author} {\bibfnamefont {W.}~\bibnamefont
  {D{\"u}r}}\ and\ \bibinfo {author} {\bibfnamefont {H.-J.}\ \bibnamefont
  {Briegel}},\ }\bibfield  {title} {\bibinfo {title} {Stability of macroscopic
  entanglement under decoherence},\ }\href
  {https://doi.org/https://doi.org/10.1103/PhysRevLett.92.180403} {\bibfield
  {journal} {\bibinfo  {journal} {Phys. Rev. Lett.}\ }\textbf {\bibinfo
  {volume} {92}},\ \bibinfo {pages} {180403} (\bibinfo {year}
  {2004})}\BibitemShut {NoStop}%
\bibitem [{\citenamefont {Carvalho}\ \emph {et~al.}(2004)\citenamefont
  {Carvalho}, \citenamefont {Mintert},\ and\ \citenamefont
  {Buchleitner}}]{carvalho2004decoherence}%
  \BibitemOpen
  \bibfield  {author} {\bibinfo {author} {\bibfnamefont {A.~R.}\ \bibnamefont
  {Carvalho}}, \bibinfo {author} {\bibfnamefont {F.}~\bibnamefont {Mintert}},\
  and\ \bibinfo {author} {\bibfnamefont {A.}~\bibnamefont {Buchleitner}},\
  }\bibfield  {title} {\bibinfo {title} {Decoherence and multipartite
  entanglement},\ }\href
  {https://doi.org/https://doi.org/10.1103/PhysRevLett.93.230501} {\bibfield
  {journal} {\bibinfo  {journal} {Phys. Rev. Lett.}\ }\textbf {\bibinfo
  {volume} {93}},\ \bibinfo {pages} {230501} (\bibinfo {year}
  {2004})}\BibitemShut {NoStop}%
\bibitem [{\citenamefont {Hein}\ \emph {et~al.}(2005)\citenamefont {Hein},
  \citenamefont {D{\"u}r},\ and\ \citenamefont
  {Briegel}}]{hein2005entanglement}%
  \BibitemOpen
  \bibfield  {author} {\bibinfo {author} {\bibfnamefont {M.}~\bibnamefont
  {Hein}}, \bibinfo {author} {\bibfnamefont {W.}~\bibnamefont {D{\"u}r}},\ and\
  \bibinfo {author} {\bibfnamefont {H.-J.}\ \bibnamefont {Briegel}},\
  }\bibfield  {title} {\bibinfo {title} {Entanglement properties of
  multipartite entangled states under the influence of decoherence},\ }\href
  {https://doi.org/https://doi.org/10.1103/PhysRevA.71.032350} {\bibfield
  {journal} {\bibinfo  {journal} {Phys. Rev. A}\ }\textbf {\bibinfo {volume}
  {71}},\ \bibinfo {pages} {032350} (\bibinfo {year} {2005})}\BibitemShut
  {NoStop}%
\bibitem [{\citenamefont {G{\"u}hne}\ \emph {et~al.}(2008)\citenamefont
  {G{\"u}hne}, \citenamefont {Bodoky},\ and\ \citenamefont
  {Blaauboer}}]{guhne2008multiparticle}%
  \BibitemOpen
  \bibfield  {author} {\bibinfo {author} {\bibfnamefont {O.}~\bibnamefont
  {G{\"u}hne}}, \bibinfo {author} {\bibfnamefont {F.}~\bibnamefont {Bodoky}},\
  and\ \bibinfo {author} {\bibfnamefont {M.}~\bibnamefont {Blaauboer}},\
  }\bibfield  {title} {\bibinfo {title} {Multiparticle entanglement under the
  influence of decoherence},\ }\href
  {https://doi.org/https://doi.org/10.1103/PhysRevA.78.060301} {\bibfield
  {journal} {\bibinfo  {journal} {Phys. Rev. A}\ }\textbf {\bibinfo {volume}
  {78}},\ \bibinfo {pages} {060301} (\bibinfo {year} {2008})}\BibitemShut
  {NoStop}%
\bibitem [{\citenamefont {Aolita}\ \emph {et~al.}(2008)\citenamefont {Aolita},
  \citenamefont {Chaves}, \citenamefont {Cavalcanti}, \citenamefont
  {Ac{\'\i}n},\ and\ \citenamefont {Davidovich}}]{aolita2008scaling}%
  \BibitemOpen
  \bibfield  {author} {\bibinfo {author} {\bibfnamefont {L.}~\bibnamefont
  {Aolita}}, \bibinfo {author} {\bibfnamefont {R.}~\bibnamefont {Chaves}},
  \bibinfo {author} {\bibfnamefont {D.}~\bibnamefont {Cavalcanti}}, \bibinfo
  {author} {\bibfnamefont {A.}~\bibnamefont {Ac{\'\i}n}},\ and\ \bibinfo
  {author} {\bibfnamefont {L.}~\bibnamefont {Davidovich}},\ }\bibfield  {title}
  {\bibinfo {title} {Scaling laws for the decay of multiqubit entanglement},\
  }\href {https://doi.org/https://doi.org/10.1103/PhysRevLett.100.080501}
  {\bibfield  {journal} {\bibinfo  {journal} {Phys. Rev. Lett.}\ }\textbf
  {\bibinfo {volume} {100}},\ \bibinfo {pages} {080501} (\bibinfo {year}
  {2008})}\BibitemShut {NoStop}%
\bibitem [{\citenamefont {Briegel}\ and\ \citenamefont
  {Raussendorf}(2001)}]{briegel2001persistent}%
  \BibitemOpen
  \bibfield  {author} {\bibinfo {author} {\bibfnamefont {H.~J.}\ \bibnamefont
  {Briegel}}\ and\ \bibinfo {author} {\bibfnamefont {R.}~\bibnamefont
  {Raussendorf}},\ }\bibfield  {title} {\bibinfo {title} {Persistent
  entanglement in arrays of interacting particles},\ }\href
  {https://doi.org/https://doi.org/10.1103/PhysRevLett.86.910} {\bibfield
  {journal} {\bibinfo  {journal} {Phys. Rev. Lett.}\ }\textbf {\bibinfo
  {volume} {86}},\ \bibinfo {pages} {910} (\bibinfo {year} {2001})}\BibitemShut
  {NoStop}%
\bibitem [{\citenamefont {Brunner}\ and\ \citenamefont
  {V{\'e}rtesi}(2012)}]{brunner2012persistency}%
  \BibitemOpen
  \bibfield  {author} {\bibinfo {author} {\bibfnamefont {N.}~\bibnamefont
  {Brunner}}\ and\ \bibinfo {author} {\bibfnamefont {T.}~\bibnamefont
  {V{\'e}rtesi}},\ }\bibfield  {title} {\bibinfo {title} {Persistency of
  entanglement and nonlocality in multipartite quantum systems},\ }\href
  {https://doi.org/https://doi.org/10.1103/PhysRevA.86.042113} {\bibfield
  {journal} {\bibinfo  {journal} {Phys. Rev. A}\ }\textbf {\bibinfo {volume}
  {86}},\ \bibinfo {pages} {042113} (\bibinfo {year} {2012})}\BibitemShut
  {NoStop}%
\bibitem [{\citenamefont {Neven}\ \emph {et~al.}(2018)\citenamefont {Neven},
  \citenamefont {Martin},\ and\ \citenamefont
  {Bastin}}]{neven2018entanglement}%
  \BibitemOpen
  \bibfield  {author} {\bibinfo {author} {\bibfnamefont {A.}~\bibnamefont
  {Neven}}, \bibinfo {author} {\bibfnamefont {J.}~\bibnamefont {Martin}},\ and\
  \bibinfo {author} {\bibfnamefont {T.}~\bibnamefont {Bastin}},\ }\bibfield
  {title} {\bibinfo {title} {Entanglement robustness against particle loss in
  multiqubit systems},\ }\href
  {https://doi.org/https://doi.org/10.1103/PhysRevA.98.062335} {\bibfield
  {journal} {\bibinfo  {journal} {Phys. Rev. A}\ }\textbf {\bibinfo {volume}
  {98}},\ \bibinfo {pages} {062335} (\bibinfo {year} {2018})}\BibitemShut
  {NoStop}%
\bibitem [{\citenamefont {Luo}\ and\ \citenamefont
  {Fei}(2021)}]{luo2021robust}%
  \BibitemOpen
  \bibfield  {author} {\bibinfo {author} {\bibfnamefont {M.-X.}\ \bibnamefont
  {Luo}}\ and\ \bibinfo {author} {\bibfnamefont {S.-M.}\ \bibnamefont {Fei}},\
  }\bibfield  {title} {\bibinfo {title} {Robust multipartite entanglement
  without entanglement breaking},\ }\href
  {https://doi.org/10.1103/PhysRevResearch.3.043120} {\bibfield  {journal}
  {\bibinfo  {journal} {Phys. Rev. Research}\ }\textbf {\bibinfo {volume}
  {3}},\ \bibinfo {pages} {043120} (\bibinfo {year} {2021})}\BibitemShut
  {NoStop}%
\bibitem [{\citenamefont {Horodecki}\ \emph {et~al.}(2005)\citenamefont
  {Horodecki}, \citenamefont {Horodecki}, \citenamefont {Horodecki},\ and\
  \citenamefont {Oppenheim}}]{horodecki2005locking}%
  \BibitemOpen
  \bibfield  {author} {\bibinfo {author} {\bibfnamefont {K.}~\bibnamefont
  {Horodecki}}, \bibinfo {author} {\bibfnamefont {M.}~\bibnamefont
  {Horodecki}}, \bibinfo {author} {\bibfnamefont {P.}~\bibnamefont
  {Horodecki}},\ and\ \bibinfo {author} {\bibfnamefont {J.}~\bibnamefont
  {Oppenheim}},\ }\bibfield  {title} {\bibinfo {title} {Locking entanglement
  with a single qubit},\ }\href
  {https://doi.org/https://doi.org/10.1103/PhysRevLett.94.200501} {\bibfield
  {journal} {\bibinfo  {journal} {Phys. Rev. Lett.}\ }\textbf {\bibinfo
  {volume} {94}},\ \bibinfo {pages} {200501} (\bibinfo {year}
  {2005})}\BibitemShut {NoStop}%
\bibitem [{\citenamefont {Christandl}\ and\ \citenamefont
  {Winter}(2005)}]{christandl2005uncertainty}%
  \BibitemOpen
  \bibfield  {author} {\bibinfo {author} {\bibfnamefont {M.}~\bibnamefont
  {Christandl}}\ and\ \bibinfo {author} {\bibfnamefont {A.}~\bibnamefont
  {Winter}},\ }\bibfield  {title} {\bibinfo {title} {Uncertainty, monogamy, and
  locking of quantum correlations},\ }\href
  {https://doi.org/https://doi.org/10.1109/TIT.2005.853338} {\bibfield
  {journal} {\bibinfo  {journal} {IEEE Trans. Inf. Theory}\ }\textbf {\bibinfo
  {volume} {51}},\ \bibinfo {pages} {3159} (\bibinfo {year}
  {2005})}\BibitemShut {NoStop}%
\bibitem [{\citenamefont {Leung}(2009)}]{leung2009survey}%
  \BibitemOpen
  \bibfield  {author} {\bibinfo {author} {\bibfnamefont {D.}~\bibnamefont
  {Leung}},\ }\bibfield  {title} {\bibinfo {title} {A survey on locking of
  bipartite correlations},\ }\href
  {https://doi.org/10.1088/1742-6596/143/1/012008} {\bibfield  {journal}
  {\bibinfo  {journal} {J. Phys. Conf. Ser.}\ }\textbf {\bibinfo {volume}
  {143}},\ \bibinfo {pages} {012008} (\bibinfo {year} {2009})}\BibitemShut
  {NoStop}%
\bibitem [{\citenamefont {Yang}\ \emph {et~al.}(2009)\citenamefont {Yang},
  \citenamefont {Horodecki}, \citenamefont {Horodecki}, \citenamefont
  {Horodecki}, \citenamefont {Oppenheim},\ and\ \citenamefont
  {Song}}]{yang2009squashed}%
  \BibitemOpen
  \bibfield  {author} {\bibinfo {author} {\bibfnamefont {D.}~\bibnamefont
  {Yang}}, \bibinfo {author} {\bibfnamefont {K.}~\bibnamefont {Horodecki}},
  \bibinfo {author} {\bibfnamefont {M.}~\bibnamefont {Horodecki}}, \bibinfo
  {author} {\bibfnamefont {P.}~\bibnamefont {Horodecki}}, \bibinfo {author}
  {\bibfnamefont {J.}~\bibnamefont {Oppenheim}},\ and\ \bibinfo {author}
  {\bibfnamefont {W.}~\bibnamefont {Song}},\ }\bibfield  {title} {\bibinfo
  {title} {Squashed entanglement for multipartite states and entanglement
  measures based on the mixed convex roof},\ }\href
  {https://doi.org/https://doi.org/10.1109/TIT.2009.2021373} {\bibfield
  {journal} {\bibinfo  {journal} {IEEE Trans. Inf. Theory}\ }\textbf {\bibinfo
  {volume} {55}},\ \bibinfo {pages} {3375} (\bibinfo {year}
  {2009})}\BibitemShut {NoStop}%
\bibitem [{\citenamefont {DiVincenzo}\ \emph {et~al.}(1998)\citenamefont
  {DiVincenzo}, \citenamefont {Fuchs}, \citenamefont {Mabuchi}, \citenamefont
  {Smolin}, \citenamefont {Thapliyal},\ and\ \citenamefont
  {Uhlmann}}]{divincenzo1998entanglement}%
  \BibitemOpen
  \bibfield  {author} {\bibinfo {author} {\bibfnamefont {D.~P.}\ \bibnamefont
  {DiVincenzo}}, \bibinfo {author} {\bibfnamefont {C.~A.}\ \bibnamefont
  {Fuchs}}, \bibinfo {author} {\bibfnamefont {H.}~\bibnamefont {Mabuchi}},
  \bibinfo {author} {\bibfnamefont {J.~A.}\ \bibnamefont {Smolin}}, \bibinfo
  {author} {\bibfnamefont {A.}~\bibnamefont {Thapliyal}},\ and\ \bibinfo
  {author} {\bibfnamefont {A.}~\bibnamefont {Uhlmann}},\ }\bibfield  {title}
  {\bibinfo {title} {Entanglement of assistance},\ }in\ \href
  {https://doi.org/https://doi.org/10.1007/3-540-49208-9_21} {\emph {\bibinfo
  {booktitle} {NASA International Conference on Quantum Computing and Quantum
  Communications}}}\ (\bibinfo {organization} {Springer},\ \bibinfo {year}
  {1998})\ pp.\ \bibinfo {pages} {247--257}\BibitemShut {NoStop}%
\bibitem [{\citenamefont {Li}\ \emph {et~al.}(2010)\citenamefont {Li},
  \citenamefont {Zhao}, \citenamefont {Fei},\ and\ \citenamefont
  {Liu}}]{li2010evolution}%
  \BibitemOpen
  \bibfield  {author} {\bibinfo {author} {\bibfnamefont {Z.-G.}\ \bibnamefont
  {Li}}, \bibinfo {author} {\bibfnamefont {M.-J.}\ \bibnamefont {Zhao}},
  \bibinfo {author} {\bibfnamefont {S.-M.}\ \bibnamefont {Fei}},\ and\ \bibinfo
  {author} {\bibfnamefont {W.}~\bibnamefont {Liu}},\ }\bibfield  {title}
  {\bibinfo {title} {Evolution equation for entanglement of assistance},\
  }\href {https://doi.org/https://doi.org/10.1103/PhysRevA.81.042312}
  {\bibfield  {journal} {\bibinfo  {journal} {Phys. Rev. A}\ }\textbf {\bibinfo
  {volume} {81}},\ \bibinfo {pages} {042312} (\bibinfo {year}
  {2010})}\BibitemShut {NoStop}%
\bibitem [{\citenamefont {Smolin}\ \emph {et~al.}(2005)\citenamefont {Smolin},
  \citenamefont {Verstraete},\ and\ \citenamefont
  {Winter}}]{smolin2005entanglement}%
  \BibitemOpen
  \bibfield  {author} {\bibinfo {author} {\bibfnamefont {J.~A.}\ \bibnamefont
  {Smolin}}, \bibinfo {author} {\bibfnamefont {F.}~\bibnamefont {Verstraete}},\
  and\ \bibinfo {author} {\bibfnamefont {A.}~\bibnamefont {Winter}},\
  }\bibfield  {title} {\bibinfo {title} {Entanglement of assistance and
  multipartite state distillation},\ }\href
  {https://doi.org/https://doi.org/10.1103/PhysRevA.72.052317} {\bibfield
  {journal} {\bibinfo  {journal} {Phys. Rev. A}\ }\textbf {\bibinfo {volume}
  {72}},\ \bibinfo {pages} {052317} (\bibinfo {year} {2005})}\BibitemShut
  {NoStop}%
\bibitem [{\citenamefont {Gour}\ \emph {et~al.}(2005)\citenamefont {Gour},
  \citenamefont {Meyer},\ and\ \citenamefont
  {Sanders}}]{gour2005deterministic}%
  \BibitemOpen
  \bibfield  {author} {\bibinfo {author} {\bibfnamefont {G.}~\bibnamefont
  {Gour}}, \bibinfo {author} {\bibfnamefont {D.~A.}\ \bibnamefont {Meyer}},\
  and\ \bibinfo {author} {\bibfnamefont {B.~C.}\ \bibnamefont {Sanders}},\
  }\bibfield  {title} {\bibinfo {title} {Deterministic entanglement of
  assistance and monogamy constraints},\ }\href
  {https://doi.org/https://doi.org/10.1103/PhysRevA.72.042329} {\bibfield
  {journal} {\bibinfo  {journal} {Phys. Rev. A}\ }\textbf {\bibinfo {volume}
  {72}},\ \bibinfo {pages} {042329} (\bibinfo {year} {2005})}\BibitemShut
  {NoStop}%
\bibitem [{\citenamefont {D'Arrigo}\ \emph {et~al.}(2014)\citenamefont
  {D'Arrigo}, \citenamefont {Benenti}, \citenamefont {Lo~Franco}, \citenamefont
  {Falci},\ and\ \citenamefont {Paladino}}]{d2014hidden}%
  \BibitemOpen
  \bibfield  {author} {\bibinfo {author} {\bibfnamefont {A.}~\bibnamefont
  {D'Arrigo}}, \bibinfo {author} {\bibfnamefont {G.}~\bibnamefont {Benenti}},
  \bibinfo {author} {\bibfnamefont {R.}~\bibnamefont {Lo~Franco}}, \bibinfo
  {author} {\bibfnamefont {G.}~\bibnamefont {Falci}},\ and\ \bibinfo {author}
  {\bibfnamefont {E.}~\bibnamefont {Paladino}},\ }\bibfield  {title} {\bibinfo
  {title} {Hidden entanglement, system-environment information flow and
  non-markovianity},\ }\href
  {https://doi.org/https://doi.org/10.1142/S021974991461005X} {\bibfield
  {journal} {\bibinfo  {journal} {InterNatl. J. (Wash.) of Quantum
  Information}\ }\textbf {\bibinfo {volume} {12}},\ \bibinfo {pages} {1461005}
  (\bibinfo {year} {2014})}\BibitemShut {NoStop}%
\bibitem [{\citenamefont {Chuan}\ \emph {et~al.}(2012)\citenamefont {Chuan},
  \citenamefont {Maillard}, \citenamefont {Modi}, \citenamefont {Paterek},
  \citenamefont {Paternostro},\ and\ \citenamefont {Piani}}]{chuan2012quantum}%
  \BibitemOpen
  \bibfield  {author} {\bibinfo {author} {\bibfnamefont {T.}~\bibnamefont
  {Chuan}}, \bibinfo {author} {\bibfnamefont {J.}~\bibnamefont {Maillard}},
  \bibinfo {author} {\bibfnamefont {K.}~\bibnamefont {Modi}}, \bibinfo {author}
  {\bibfnamefont {T.}~\bibnamefont {Paterek}}, \bibinfo {author} {\bibfnamefont
  {M.}~\bibnamefont {Paternostro}},\ and\ \bibinfo {author} {\bibfnamefont
  {M.}~\bibnamefont {Piani}},\ }\bibfield  {title} {\bibinfo {title} {Quantum
  discord bounds the amount of distributed entanglement},\ }\href
  {https://doi.org/https://doi.org/10.1103/PhysRevLett.109.070501} {\bibfield
  {journal} {\bibinfo  {journal} {Phys. Rev. Lett.}\ }\textbf {\bibinfo
  {volume} {109}},\ \bibinfo {pages} {070501} (\bibinfo {year}
  {2012})}\BibitemShut {NoStop}%
\bibitem [{\citenamefont {Streltsov}\ \emph {et~al.}(2012)\citenamefont
  {Streltsov}, \citenamefont {Kampermann},\ and\ \citenamefont
  {Bru{\ss}}}]{streltsov2012quantum}%
  \BibitemOpen
  \bibfield  {author} {\bibinfo {author} {\bibfnamefont {A.}~\bibnamefont
  {Streltsov}}, \bibinfo {author} {\bibfnamefont {H.}~\bibnamefont
  {Kampermann}},\ and\ \bibinfo {author} {\bibfnamefont {D.}~\bibnamefont
  {Bru{\ss}}},\ }\bibfield  {title} {\bibinfo {title} {Quantum cost for sending
  entanglement},\ }\href
  {https://doi.org/https://doi.org/10.1103/PhysRevLett.108.250501} {\bibfield
  {journal} {\bibinfo  {journal} {Phys. Rev. Lett.}\ }\textbf {\bibinfo
  {volume} {108}},\ \bibinfo {pages} {250501} (\bibinfo {year}
  {2012})}\BibitemShut {NoStop}%
\bibitem [{\citenamefont {Horodecki}\ \emph {et~al.}(2003)\citenamefont
  {Horodecki}, \citenamefont {Shor},\ and\ \citenamefont
  {Ruskai}}]{horodecki2003entanglement}%
  \BibitemOpen
  \bibfield  {author} {\bibinfo {author} {\bibfnamefont {M.}~\bibnamefont
  {Horodecki}}, \bibinfo {author} {\bibfnamefont {P.~W.}\ \bibnamefont
  {Shor}},\ and\ \bibinfo {author} {\bibfnamefont {M.~B.}\ \bibnamefont
  {Ruskai}},\ }\bibfield  {title} {\bibinfo {title} {Entanglement breaking
  channels},\ }\href
  {https://doi.org/https://doi.org/10.1142/S0129055X03001709} {\bibfield
  {journal} {\bibinfo  {journal} {Rev. Math. Phys.}\ }\textbf {\bibinfo
  {volume} {15}},\ \bibinfo {pages} {629} (\bibinfo {year} {2003})}\BibitemShut
  {NoStop}%
\bibitem [{\citenamefont {Chitambar}\ \emph {et~al.}(2014)\citenamefont
  {Chitambar}, \citenamefont {Leung}, \citenamefont {Man{\v{c}}inska},
  \citenamefont {Ozols},\ and\ \citenamefont
  {Winter}}]{chitambar2014everything}%
  \BibitemOpen
  \bibfield  {author} {\bibinfo {author} {\bibfnamefont {E.}~\bibnamefont
  {Chitambar}}, \bibinfo {author} {\bibfnamefont {D.}~\bibnamefont {Leung}},
  \bibinfo {author} {\bibfnamefont {L.}~\bibnamefont {Man{\v{c}}inska}},
  \bibinfo {author} {\bibfnamefont {M.}~\bibnamefont {Ozols}},\ and\ \bibinfo
  {author} {\bibfnamefont {A.}~\bibnamefont {Winter}},\ }\bibfield  {title}
  {\bibinfo {title} {Everything you always wanted to know about locc (but were
  afraid to ask)},\ }\href
  {https://doi.org/https://doi.org/10.1007/s00220-014-1953-9} {\bibfield
  {journal} {\bibinfo  {journal} {Commun. Math. Phys.}\ }\textbf {\bibinfo
  {volume} {328}},\ \bibinfo {pages} {303} (\bibinfo {year}
  {2014})}\BibitemShut {NoStop}%
\bibitem [{sup()}]{supplementalmaterial}%
  \BibitemOpen
  \href@noop {} {\bibinfo {title} {{See Supplemental Material} at [url will be
  inserted by publisher] for the technical details.}}\BibitemShut {Stop}%
\bibitem [{\citenamefont {D'Ariano}\ \emph {et~al.}(2005)\citenamefont
  {D'Ariano}, \citenamefont {Presti},\ and\ \citenamefont
  {Perinotti}}]{d2005classical}%
  \BibitemOpen
  \bibfield  {author} {\bibinfo {author} {\bibfnamefont {G.~M.}\ \bibnamefont
  {D'Ariano}}, \bibinfo {author} {\bibfnamefont {P.~L.}\ \bibnamefont
  {Presti}},\ and\ \bibinfo {author} {\bibfnamefont {P.}~\bibnamefont
  {Perinotti}},\ }\bibfield  {title} {\bibinfo {title} {Classical randomness in
  quantum measurements},\ }\href
  {https://doi.org/https://doi.org/10.1088/0305-4470/38/26/010} {\bibfield
  {journal} {\bibinfo  {journal} {J. Phys. A: Math. Gen.}\ }\textbf {\bibinfo
  {volume} {38}},\ \bibinfo {pages} {5979} (\bibinfo {year}
  {2005})}\BibitemShut {NoStop}%
\bibitem [{\citenamefont {Ruskai}(2003)}]{ruskai2003qubit}%
  \BibitemOpen
  \bibfield  {author} {\bibinfo {author} {\bibfnamefont {M.~B.}\ \bibnamefont
  {Ruskai}},\ }\bibfield  {title} {\bibinfo {title} {Qubit entanglement
  breaking channels},\ }\href
  {https://doi.org/https://doi.org/10.1142/S0129055X03001710} {\bibfield
  {journal} {\bibinfo  {journal} {Rev. Math. Phys.}\ }\textbf {\bibinfo
  {volume} {15}},\ \bibinfo {pages} {643} (\bibinfo {year} {2003})}\BibitemShut
  {NoStop}%
\bibitem [{\citenamefont {Sab{\'\i}n}\ and\ \citenamefont
  {Garc{\'\i}a-Alcaine}(2008)}]{sabin2008classification}%
  \BibitemOpen
  \bibfield  {author} {\bibinfo {author} {\bibfnamefont {C.}~\bibnamefont
  {Sab{\'\i}n}}\ and\ \bibinfo {author} {\bibfnamefont {G.}~\bibnamefont
  {Garc{\'\i}a-Alcaine}},\ }\bibfield  {title} {\bibinfo {title} {A
  classification of entanglement in three-qubit systems},\ }\href
  {https://doi.org/https://doi.org/10.1140/epjd/e2008-00112-5} {\bibfield
  {journal} {\bibinfo  {journal} {Eur. Phys. J. D}\ }\textbf {\bibinfo {volume}
  {48}},\ \bibinfo {pages} {435} (\bibinfo {year} {2008})}\BibitemShut
  {NoStop}%
\bibitem [{\citenamefont {Christandl}\ and\ \citenamefont
  {Winter}(2004)}]{christandl2004squashed}%
  \BibitemOpen
  \bibfield  {author} {\bibinfo {author} {\bibfnamefont {M.}~\bibnamefont
  {Christandl}}\ and\ \bibinfo {author} {\bibfnamefont {A.}~\bibnamefont
  {Winter}},\ }\bibfield  {title} {\bibinfo {title} {{“Squashed
  entanglement”}: an additive entanglement measure},\ }\href
  {https://doi.org/https://doi.org/10.1063/1.1643788} {\bibfield  {journal}
  {\bibinfo  {journal} {J. Math. Phys.}\ }\textbf {\bibinfo {volume} {45}},\
  \bibinfo {pages} {829} (\bibinfo {year} {2004})}\BibitemShut {NoStop}%
\bibitem [{\citenamefont {Linden}\ \emph {et~al.}(1999)\citenamefont {Linden},
  \citenamefont {Popescu}, \citenamefont {Schumacher},\ and\ \citenamefont
  {Westmoreland}}]{linden1999reversibility}%
  \BibitemOpen
  \bibfield  {author} {\bibinfo {author} {\bibfnamefont {N.}~\bibnamefont
  {Linden}}, \bibinfo {author} {\bibfnamefont {S.}~\bibnamefont {Popescu}},
  \bibinfo {author} {\bibfnamefont {B.}~\bibnamefont {Schumacher}},\ and\
  \bibinfo {author} {\bibfnamefont {M.}~\bibnamefont {Westmoreland}},\
  }\bibfield  {title} {\bibinfo {title} {Reversibility of local transformations
  of multiparticle entanglement},\ }\href@noop {} {\bibfield  {journal}
  {\bibinfo  {journal} {arXiv: quant-ph/9912039}\ } (\bibinfo {year}
  {1999})}\BibitemShut {NoStop}%
\bibitem [{\citenamefont {Modi}\ \emph {et~al.}(2012)\citenamefont {Modi},
  \citenamefont {Brodutch}, \citenamefont {Cable}, \citenamefont {Paterek},\
  and\ \citenamefont {Vedral}}]{modi2012classical}%
  \BibitemOpen
  \bibfield  {author} {\bibinfo {author} {\bibfnamefont {K.}~\bibnamefont
  {Modi}}, \bibinfo {author} {\bibfnamefont {A.}~\bibnamefont {Brodutch}},
  \bibinfo {author} {\bibfnamefont {H.}~\bibnamefont {Cable}}, \bibinfo
  {author} {\bibfnamefont {T.}~\bibnamefont {Paterek}},\ and\ \bibinfo {author}
  {\bibfnamefont {V.}~\bibnamefont {Vedral}},\ }\bibfield  {title} {\bibinfo
  {title} {The classical-quantum boundary for correlations: Discord and related
  measures},\ }\href
  {https://doi.org/https://doi.org/10.1103/RevModPhys.84.1655} {\bibfield
  {journal} {\bibinfo  {journal} {Rev. Mod. Phys.}\ }\textbf {\bibinfo {volume}
  {84}},\ \bibinfo {pages} {1655} (\bibinfo {year} {2012})}\BibitemShut
  {NoStop}%
\bibitem [{\citenamefont {Horodecki}(2005)}]{horodecki2005simplifying}%
  \BibitemOpen
  \bibfield  {author} {\bibinfo {author} {\bibfnamefont {M.}~\bibnamefont
  {Horodecki}},\ }\bibfield  {title} {\bibinfo {title} {Simplifying
  monotonicity conditions for entanglement measures},\ }\href
  {https://doi.org/https://doi.org/10.1007/s11080-005-0920-5} {\bibfield
  {journal} {\bibinfo  {journal} {Open Systems \& Information Dynamics}\
  }\textbf {\bibinfo {volume} {12}},\ \bibinfo {pages} {231} (\bibinfo {year}
  {2005})}\BibitemShut {NoStop}%
\bibitem [{\citenamefont {Miklin}\ \emph {et~al.}(2016)\citenamefont {Miklin},
  \citenamefont {Moroder},\ and\ \citenamefont
  {G{\"u}hne}}]{miklin2016multiparticle}%
  \BibitemOpen
  \bibfield  {author} {\bibinfo {author} {\bibfnamefont {N.}~\bibnamefont
  {Miklin}}, \bibinfo {author} {\bibfnamefont {T.}~\bibnamefont {Moroder}},\
  and\ \bibinfo {author} {\bibfnamefont {O.}~\bibnamefont {G{\"u}hne}},\
  }\bibfield  {title} {\bibinfo {title} {Multiparticle entanglement as an
  emergent phenomenon},\ }\href
  {https://doi.org/https://doi.org/10.1103/PhysRevA.93.020104} {\bibfield
  {journal} {\bibinfo  {journal} {Phys. Rev. A}\ }\textbf {\bibinfo {volume}
  {93}},\ \bibinfo {pages} {020104} (\bibinfo {year} {2016})}\BibitemShut
  {NoStop}%
\bibitem [{\citenamefont {Seshadreesan}\ and\ \citenamefont
  {Wilde}(2015)}]{seshadreesan2015fidelity}%
  \BibitemOpen
  \bibfield  {author} {\bibinfo {author} {\bibfnamefont {K.~P.}\ \bibnamefont
  {Seshadreesan}}\ and\ \bibinfo {author} {\bibfnamefont {M.~M.}\ \bibnamefont
  {Wilde}},\ }\bibfield  {title} {\bibinfo {title} {Fidelity of recovery,
  squashed entanglement, and measurement recoverability},\ }\href
  {https://doi.org/https://doi.org/10.1103/PhysRevA.92.042321} {\bibfield
  {journal} {\bibinfo  {journal} {Phys. Rev. A}\ }\textbf {\bibinfo {volume}
  {92}},\ \bibinfo {pages} {042321} (\bibinfo {year} {2015})}\BibitemShut
  {NoStop}%
\bibitem [{\citenamefont {Bennett}\ \emph {et~al.}(1999)\citenamefont
  {Bennett}, \citenamefont {DiVincenzo}, \citenamefont {Mor}, \citenamefont
  {Shor}, \citenamefont {Smolin},\ and\ \citenamefont
  {Terhal}}]{bennett1999unextendible}%
  \BibitemOpen
  \bibfield  {author} {\bibinfo {author} {\bibfnamefont {C.~H.}\ \bibnamefont
  {Bennett}}, \bibinfo {author} {\bibfnamefont {D.~P.}\ \bibnamefont
  {DiVincenzo}}, \bibinfo {author} {\bibfnamefont {T.}~\bibnamefont {Mor}},
  \bibinfo {author} {\bibfnamefont {P.~W.}\ \bibnamefont {Shor}}, \bibinfo
  {author} {\bibfnamefont {J.~A.}\ \bibnamefont {Smolin}},\ and\ \bibinfo
  {author} {\bibfnamefont {B.~M.}\ \bibnamefont {Terhal}},\ }\bibfield  {title}
  {\bibinfo {title} {Unextendible product bases and bound entanglement},\
  }\href {https://doi.org/https://doi.org/10.1103/PhysRevLett.82.5385}
  {\bibfield  {journal} {\bibinfo  {journal} {Phys. Rev. Lett.}\ }\textbf
  {\bibinfo {volume} {82}},\ \bibinfo {pages} {5385} (\bibinfo {year}
  {1999})}\BibitemShut {NoStop}%
\bibitem [{\citenamefont {Acin}\ \emph {et~al.}(2001)\citenamefont {Acin},
  \citenamefont {Bru{\ss}}, \citenamefont {Lewenstein},\ and\ \citenamefont
  {Sanpera}}]{acin2001classification}%
  \BibitemOpen
  \bibfield  {author} {\bibinfo {author} {\bibfnamefont {A.}~\bibnamefont
  {Acin}}, \bibinfo {author} {\bibfnamefont {D.}~\bibnamefont {Bru{\ss}}},
  \bibinfo {author} {\bibfnamefont {M.}~\bibnamefont {Lewenstein}},\ and\
  \bibinfo {author} {\bibfnamefont {A.}~\bibnamefont {Sanpera}},\ }\bibfield
  {title} {\bibinfo {title} {Classification of mixed three-qubit states},\
  }\href {https://doi.org/https://doi.org/10.1103/PhysRevLett.87.040401}
  {\bibfield  {journal} {\bibinfo  {journal} {Phys. Rev. Lett.}\ }\textbf
  {\bibinfo {volume} {87}},\ \bibinfo {pages} {040401} (\bibinfo {year}
  {2001})}\BibitemShut {NoStop}%
\bibitem [{\citenamefont {Kraus}\ \emph {et~al.}(2000)\citenamefont {Kraus},
  \citenamefont {Cirac}, \citenamefont {Karnas},\ and\ \citenamefont
  {Lewenstein}}]{kraus2000separability}%
  \BibitemOpen
  \bibfield  {author} {\bibinfo {author} {\bibfnamefont {B.}~\bibnamefont
  {Kraus}}, \bibinfo {author} {\bibfnamefont {J.}~\bibnamefont {Cirac}},
  \bibinfo {author} {\bibfnamefont {S.}~\bibnamefont {Karnas}},\ and\ \bibinfo
  {author} {\bibfnamefont {M.}~\bibnamefont {Lewenstein}},\ }\bibfield  {title}
  {\bibinfo {title} {Separability in {2$\times$ N} composite quantum systems},\
  }\href {https://doi.org/https://doi.org/10.1103/PhysRevA.61.062302}
  {\bibfield  {journal} {\bibinfo  {journal} {Phys. Rev. A}\ }\textbf {\bibinfo
  {volume} {61}},\ \bibinfo {pages} {062302} (\bibinfo {year}
  {2000})}\BibitemShut {NoStop}%
\bibitem [{\citenamefont {Horodecki}(1997)}]{horodecki1997separability}%
  \BibitemOpen
  \bibfield  {author} {\bibinfo {author} {\bibfnamefont {P.}~\bibnamefont
  {Horodecki}},\ }\bibfield  {title} {\bibinfo {title} {Separability criterion
  and inseparable mixed states with positive partial transposition},\ }\href
  {https://doi.org/https://doi.org/10.1016/S0375-9601(97)00416-7} {\bibfield
  {journal} {\bibinfo  {journal} {Phys. Lett. A}\ }\textbf {\bibinfo {volume}
  {232}},\ \bibinfo {pages} {333} (\bibinfo {year} {1997})}\BibitemShut
  {NoStop}%
\bibitem [{\citenamefont {Kay}(2011)}]{kay2011optimal}%
  \BibitemOpen
  \bibfield  {author} {\bibinfo {author} {\bibfnamefont {A.}~\bibnamefont
  {Kay}},\ }\bibfield  {title} {\bibinfo {title} {Optimal detection of
  entanglement in {Greenberger-Horne-Zeilinger} states},\ }\href
  {https://doi.org/https://doi.org/10.1103/PhysRevA.83.020303} {\bibfield
  {journal} {\bibinfo  {journal} {Phys. Rev. A}\ }\textbf {\bibinfo {volume}
  {83}},\ \bibinfo {pages} {020303} (\bibinfo {year} {2011})}\BibitemShut
  {NoStop}%
\bibitem [{\citenamefont {Hyllus}(2005)}]{hyllus2005witnessing}%
  \BibitemOpen
  \bibfield  {author} {\bibinfo {author} {\bibfnamefont {P.}~\bibnamefont
  {Hyllus}},\ }\emph {\bibinfo {title} {Witnessing entanglement in qudit
  systems}},\ \href@noop {} {Ph.D. thesis},\ \bibinfo  {school} {Hannover:
  Universit{\"a}t} (\bibinfo {year} {2005})\BibitemShut {NoStop}%
\bibitem [{\citenamefont {G{\"u}hne}(2011)}]{guhne2011entanglement}%
  \BibitemOpen
  \bibfield  {author} {\bibinfo {author} {\bibfnamefont {O.}~\bibnamefont
  {G{\"u}hne}},\ }\bibfield  {title} {\bibinfo {title} {Entanglement criteria
  and full separability of multi-qubit quantum states},\ }\href
  {https://doi.org/https://doi.org/10.1016/j.physleta.2010.11.032} {\bibfield
  {journal} {\bibinfo  {journal} {Phys. Lett. A}\ }\textbf {\bibinfo {volume}
  {375}},\ \bibinfo {pages} {406} (\bibinfo {year} {2011})}\BibitemShut
  {NoStop}%
\bibitem [{\citenamefont {Nagata}(2009)}]{nagata2009necessary}%
  \BibitemOpen
  \bibfield  {author} {\bibinfo {author} {\bibfnamefont {K.}~\bibnamefont
  {Nagata}},\ }\bibfield  {title} {\bibinfo {title} {Necessary and sufficient
  condition for {Greenberger-Horne-Zeilinger} diagonal states to be full
  n-partite entangled},\ }\href
  {https://doi.org/https://doi.org/10.1007/s10773-009-0139-2} {\bibfield
  {journal} {\bibinfo  {journal} {Int. J. Theor. Phys.}\ }\textbf {\bibinfo
  {volume} {48}},\ \bibinfo {pages} {3358} (\bibinfo {year}
  {2009})}\BibitemShut {NoStop}%
\bibitem [{\citenamefont {Horodecki}\ \emph {et~al.}(2006)\citenamefont
  {Horodecki}, \citenamefont {Horodecki},\ and\ \citenamefont
  {Horodecki}}]{horodecki2006separability}%
  \BibitemOpen
  \bibfield  {author} {\bibinfo {author} {\bibfnamefont {M.}~\bibnamefont
  {Horodecki}}, \bibinfo {author} {\bibfnamefont {P.}~\bibnamefont
  {Horodecki}},\ and\ \bibinfo {author} {\bibfnamefont {R.}~\bibnamefont
  {Horodecki}},\ }\bibfield  {title} {\bibinfo {title} {Separability of mixed
  quantum states: linear contractions and permutation criteria},\ }\href
  {https://doi.org/https://doi.org/10.1007/s11080-006-7271-8} {\bibfield
  {journal} {\bibinfo  {journal} {Open Systems \& Information Dynamics}\
  }\textbf {\bibinfo {volume} {13}},\ \bibinfo {pages} {103} (\bibinfo {year}
  {2006})}\BibitemShut {NoStop}%
\bibitem [{\citenamefont {G{\"u}hne}\ and\ \citenamefont
  {Seevinck}(2010)}]{guhne2010separability}%
  \BibitemOpen
  \bibfield  {author} {\bibinfo {author} {\bibfnamefont {O.}~\bibnamefont
  {G{\"u}hne}}\ and\ \bibinfo {author} {\bibfnamefont {M.}~\bibnamefont
  {Seevinck}},\ }\bibfield  {title} {\bibinfo {title} {Separability criteria
  for genuine multiparticle entanglement},\ }\href
  {https://doi.org/https://doi.org/10.1088/1367-2630/12/5/053002} {\bibfield
  {journal} {\bibinfo  {journal} {New J. Phys.}\ }\textbf {\bibinfo {volume}
  {12}},\ \bibinfo {pages} {053002} (\bibinfo {year} {2010})}\BibitemShut
  {NoStop}%
\bibitem [{\citenamefont {Eggeling}\ and\ \citenamefont
  {Werner}(2001)}]{eggeling2001separability}%
  \BibitemOpen
  \bibfield  {author} {\bibinfo {author} {\bibfnamefont {T.}~\bibnamefont
  {Eggeling}}\ and\ \bibinfo {author} {\bibfnamefont {R.~F.}\ \bibnamefont
  {Werner}},\ }\bibfield  {title} {\bibinfo {title} {Separability properties of
  tripartite states with {$U\otimes U\otimes U$} symmetry},\ }\href
  {https://doi.org/https://doi.org/10.1103/PhysRevA.63.042111} {\bibfield
  {journal} {\bibinfo  {journal} {Phys. Rev. A}\ }\textbf {\bibinfo {volume}
  {63}},\ \bibinfo {pages} {042111} (\bibinfo {year} {2001})}\BibitemShut
  {NoStop}%
\bibitem [{\citenamefont {T{\'o}th}\ \emph {et~al.}(2007)\citenamefont
  {T{\'o}th}, \citenamefont {Knapp}, \citenamefont {G{\"u}hne},\ and\
  \citenamefont {Briegel}}]{toth2007optimal}%
  \BibitemOpen
  \bibfield  {author} {\bibinfo {author} {\bibfnamefont {G.}~\bibnamefont
  {T{\'o}th}}, \bibinfo {author} {\bibfnamefont {C.}~\bibnamefont {Knapp}},
  \bibinfo {author} {\bibfnamefont {O.}~\bibnamefont {G{\"u}hne}},\ and\
  \bibinfo {author} {\bibfnamefont {H.~J.}\ \bibnamefont {Briegel}},\
  }\bibfield  {title} {\bibinfo {title} {Optimal spin squeezing inequalities
  detect bound entanglement in spin models},\ }\href
  {https://doi.org/https://doi.org/10.1103/PhysRevLett.99.250405} {\bibfield
  {journal} {\bibinfo  {journal} {Phys. Rev. Lett.}\ }\textbf {\bibinfo
  {volume} {99}},\ \bibinfo {pages} {250405} (\bibinfo {year}
  {2007})}\BibitemShut {NoStop}%
\bibitem [{\citenamefont {T{\'o}th}\ \emph {et~al.}(2009)\citenamefont
  {T{\'o}th}, \citenamefont {Knapp}, \citenamefont {G{\"u}hne},\ and\
  \citenamefont {Briegel}}]{toth2009spin}%
  \BibitemOpen
  \bibfield  {author} {\bibinfo {author} {\bibfnamefont {G.}~\bibnamefont
  {T{\'o}th}}, \bibinfo {author} {\bibfnamefont {C.}~\bibnamefont {Knapp}},
  \bibinfo {author} {\bibfnamefont {O.}~\bibnamefont {G{\"u}hne}},\ and\
  \bibinfo {author} {\bibfnamefont {H.~J.}\ \bibnamefont {Briegel}},\
  }\bibfield  {title} {\bibinfo {title} {Spin squeezing and entanglement},\
  }\href {https://doi.org/https://doi.org/10.1103/PhysRevA.79.042334}
  {\bibfield  {journal} {\bibinfo  {journal} {Phys. Rev. A}\ }\textbf {\bibinfo
  {volume} {79}},\ \bibinfo {pages} {042334} (\bibinfo {year}
  {2009})}\BibitemShut {NoStop}%
\end{thebibliography}%

\end{document}